\documentclass[12pt]{article}
\usepackage{fullpage}
\usepackage[english]{babel}
\usepackage{amsmath,amsthm,amssymb}
\usepackage{algorithm}
\usepackage{amsfonts}
\usepackage[flushleft]{threeparttable}
\usepackage{relsize}
\usepackage{natbib}
 \bibpunct[, ]{(}{)}{,}{a}{}{,}%
 \def\bibfont{\small}%
 \def\bibsep{\smallskipamount}%
 \def\bibhang{24pt}%
\newtheorem{thm}{Theorem}[section]

\newtheorem{prop}{Proposition}
\newtheorem{defn}{Definition}
\newtheorem{rem}{Remark}
\newtheorem{corollary}{Corollary}

\numberwithin{alg}{section}
\numberwithin{defn}{section}
\numberwithin{equation}{section}
\numberwithin{rem}{section}
\numberwithin{corollary}{section}
\numberwithin{table}{section}
\numberwithin{prop}{section}

\begin{document}

\title{Fast Estimation of True Bounds on Bermudan Option Prices Under Jump-diffusion Processes }
\author{Helin Zhu \quad Fan Ye \quad Enlu Zhou\\
Department of Industrial and Enterprise Systems Engineering\\
University of Illinois}
\date{\today}
\maketitle
\begin{abstract}
Fast pricing of American-style options has been a difficult problem since it was first introduced to financial markets in 1970s, especially when the underlying stocks' prices follow some jump-diffusion processes. In this paper, we propose a new algorithm to generate tight upper bounds on the Bermudan option price without nested simulation, under the jump-diffusion setting. By exploiting the martingale representation theorem for jump processes on the dual martingale, we are able to explore the unique structure of the optimal dual martingale and construct an approximation that preserves the martingale property. The resulting upper bound estimator avoids the nested Monte Carlo simulation suffered by the original primal-dual algorithm, therefore significantly improves the computational efficiency.  Theoretical analysis is provided to guarantee the quality of the martingale approximation.  Numerical experiments are conducted to verify the efficiency of our proposed algorithm.
\end{abstract}
\newpage
\section{Introduction}
Pricing American-style derivatives (which is essentially an optimal stopping problem) has been an active and challenging problem in the last thirty years, especially when the underlying stocks' prices follow some jump-diffusion processes, as they become more and more critical to investors. To present time, various jump-diffusion models for financial modelling have been proposed to fit the real data in financial markets, including: (i) the normal jump-diffusion model, see \cite{merton1976option}; (ii) the affine jump-diffusion models, see \cite{duffie2000transform}; (iii) the jump models based on Levy processes, see \cite{cont2003financial}; (iv) the double exponential, mixed-exponential and hyper-exponential jump-diffusion models, see \cite{kou2002jump}, \cite{cai2011option}, and \cite{cai2012option}. All these models are trying to capture some interesting features of the market behaviour that cannot be well explained by pure-diffusion models, such as the heavy-tail risk suffered by the market. In general, closed-form expressions for the American-style derivatives can hardly be derived under these jump-diffusion models due to the multiple exercise opportunities and the randomness in the underlying asset price caused by both jumps and diffusions. Hence, various numerical methods have been proposed to tackle the American-style option pricing problems under the jump-diffusion models, including: (i) solving the free boundary problems via lattice or differential equation methods, see \cite{amin1993jump}, \cite{kellezi2004valuing}, \cite{feng2008pricing}, \cite{fang2009pricing}, \cite{Feng2012hilbert}; (ii) quadratic approximation and piece-wise exponential approximation methods, see \cite{pham1997optimal}, \cite{gukhal2001analytical}, \cite{kou2004option}. A thorough study on jump-diffusion models for asset pricing has been done by \cite{kou2008jump}. More broadly, an elegant overview of financial models under jump processes is provided in \cite{cont2003financial}.\\
\\
Another class of widely-used methods are based on Monte Carlo simulation, and they have been successfully implemented on Bermudan option pricing problems under the pure-diffusion models, see \cite{bossaerts1989simulation}, \cite{tilley:1993}, \cite{longstaff:2001}, \cite{tsitsiklis:2001}. In particular,  \cite{longstaff:2001}, \cite{tsitsiklis:2001} propose to approximate the continuation values by regression on certain basis functions set (called ``function basis''), which leads to good suboptimal exercise strategies and lower bounds on the exact option price. Moreover, their methods bypass ``the curse of dimensionality'' and scale well with the number of underlying variables, working efficiently for high-dimensional problems under the pure-diffusion models. Though these methods can be naturally adapted to option pricing problems under the jump-diffusion setting, two key questions regarding the effectiveness of these methods remain to be addressed: (i) how to choose the function bases for regression. (ii) how to measure the quality of the lower bounds. \\
\\
The second question is partially addressed by the dual approach proposed independently by \cite{rogers:2002}, \cite{haugh:2004}, and \cite{andersen:2004}. They are able to generate the upper bounds on the option price by solving the associated dual problem, which is obtained by subtracting the payoff function by a dual martingale adapted to a proper filtration. In theory, if the dual martingale is the Doob-Meyer martingale part of the option price process, namely the ``optimal dual martingale'', then the resulting upper bound equals the exact option price. In practice, the optimal dual martingale is not available, but good approximations of it can generate tight upper bounds.  With the access to the upper bounds, the quality of suboptimal exercise strategies or lower bounds could be measured empirically by looking at the duality gaps, which are the differences between the lower bounds and the upper bounds. A multiplicative version of dual approach based on multiplicative Doob-Meyer decomposition is proposed by \cite{Jamshidian:2006}. A thorough comparison between the additive dual approach and the multiplicative dual approach can be found in \cite{Chen:2007}.  \cite{glasserman:2004} provides an elegant and thorough overview on the duality theory for option pricing problems.\\
\\
A lot work has been done following the duality theory. To name a few, \cite{bender2011dual}, \cite{chandramouli:2012}, and \cite{Bender:2011} develop the multilevel primal-dual approach for optimal stopping problems with multiple stopping times. \cite{belomestny2012multilevel} optimize the cost of simulation by considering a multilevel Monte Carlo technique for the primal-dual approach. \cite{desai2012pathwise} consider an additional path-wise optimization procedure in constructing the dual martingales for optimal stopping problems. \cite{ye2013pricing} apply the primal-dual approach with particle filtering techniques to optimal stopping problems of partially observable Markov processes. \cite{rogers:2007}, \cite{brown:2010} generalize the duality theory to general discrete-time dynamic programming problems and provide a broader interpretation of the dual martingale. From \cite{brown:2010}'s perspective, the dual martingale can be regarded as the penalty for the access to the future information (information relaxation) and different degrees of relaxation result in different levels of upper bounds. In particular, the dual martingales constructed by \cite{haugh:2004}, \cite{andersen:2004} can be interpreted as perfect information relaxation, which means the option holder has access to all the future prices of the underlying assets. \cite{ye2012parameterized} consider a parameterized path-wise optimization technique in constructing the penalties for general dynamic programming problems. \cite{Ye2013} also develop the duality theory for general dynamic programming problems under a continuous-time setting.\\
\\
The numerical effectiveness of the primal-dual approach has been demonstrated in pricing multidimensional American-style options. The algorithm generates good suboptimal exercise strategies and good lower-upper bound pairs with sufficiently small duality gaps. A possible deficiency of the algorithm is the heavy computation (quadratic in the number of exercisable periods), caused by the nested simulation in constructing the dual martingale. To address the computational issue, \cite{Belomestny:2009} propose an alternative algorithm to generate approximations of the optimal dual martingale via non-nested simulation under the Wiener process setting. By exploiting the martingale representation theorem on the optimal dual martingale driven by Wiener processes, they are able to approximate the optimal dual martingale through regressing the integrand on some function bases at finite number of time points. The resulting approximation preserves the martingale property and generates a true upper bound on the option price. More importantly, their algorithm avoids the nested simulation and is linear in the number of exercisable periods.\\
\\
In this paper, we will generalize \cite{Belomestny:2009}'s idea of ``true martingale'' to Bermudan option pricing problems under jump-diffusion processes and provide a new perspective in understanding the structure of the optimal dual martingale, which facilitates us to construct good approximations of it. According to our knowledge, we are among the first to ever consider estimating the upper bounds on American-style option price under the jump-diffusion models. In a greater detail, we have made the following contributions.
\begin{itemize}
\item Under the jump-diffusion setting, we explore the structure of the optimal dual martingales in the dual representation of both the Bermudan and American option prices (Theorem \ref{thm:martingale} and Theorem \ref{thm:derivative}), which is the underpinning of our proposed approach of generating tight upper bounds.
\item Motivated by \cite{Belomestny:2009}, we propose a new algorithm, which is referred as the "true martingale algorithm" (T-M algorithm), to compute the upper bounds on the Bermudan option price under the jump-diffusion models.  The resulting approximation (called ``true martingale approximation'') preserves the martingale property, and therefore generates true upper bounds on the Bermudan option price. Moreover, compared with the primal-dual algorithm proposed by \cite{andersen:2004} (A-B algorithm), our proposed T-M algorithm avoids the nested Monte Carlo simulation and scales linearly with the exercisable periods, and hence achieves a higher computational efficiency.
\item We prove that the true martingale approximation converges to the objective martingale in the mean square sense provided that the time discretization goes to zero by bounding the empirical difference between the approximation and the objective martingale (Theorem \ref{thm:3:4}).
\item We investigate the numerical effectiveness of \cite{longstaff:2001}'s least-squares regression approach (L-S algorithm) for Bermudan option price under the jump-diffusion models. In particular, we find that by incorporating the European option price under the corresponding pure-diffusion model (referred as the ``non-jump European option'') into the function basis of the L-S algorithm, the quality of the induced suboptimal exercise strategies and the lower bounds can be significantly improved.
\item Motivated by the explicit structure of the optimal dual martingale (Theorem \ref{thm:martingale}), we propose a function basis that can be employed in our proposed algorithm to obtain upper bounds on the option price.  This function basis is also derived based on the non-jump European option price, which is critical to the true martingale approximation and hence the quality of the true upper bounds. By implementing our algorithm together with the A-B algorithm on several sets of numerical experiments, the numerical results  demonstrate that both methods can generate tight and stable upper bounds on option price of the same quality; however, we observe that our algorithm is much more efficient than the A-B algorithm in practice due to the relief from nested simulation.
\end{itemize}
To summarize, the rest of this paper will be organized as follows. In section 2, we describe the Bermudan option pricing problems under general jump-diffusion models and review the dual approach. We develop the true martingale approach and provide error analysis and convergence analysis of it in section 3. Section 4 focuses on the detailed T-M algorithm and its numerical advantages.  Numerical experiments are conducted in section 5 to verify the computational efficiency of the T-M algorithm. Conclusion and future directions are given in section 6.

\section{Model Formulation}

\subsection{Preliminaries}
In this subsection, we will introduce some standard definitions in jump processes that will appear throughout the paper. They can be found in \cite{cont2003financial}.
\begin{defn}[\textbf{Poisson random measure} (Definition 2.18 in \cite{cont2003financial}]\label{def:2.1}
Let $(\Omega, \mathcal{F}, \mathbb{P})$ be a probability space, $G \subset \mathbb{R}^{d+1} $ and $\mu$ be a given (positive) Radon measure on $(G, \mathcal{G})$. A Poisson random measure on $G$ with intensity $\mu$ is an integer-valued random measure:
\begin{equation*}
\begin{split}\mathcal{P}:
& \Omega \times \mathcal{G} \rightarrow \mathbb{N}\\
& \left(\omega, A\right)\mapsto \mathcal{P}\left(\omega, A\right)
\end{split}
\end{equation*}
such that: (i) For (almost all) $\omega \in \Omega, \mathcal{P}(\omega, \cdot)$ is an integer-valued Radon measure on $G$; (ii) for each measurable set $A\in \mathcal{G}$, $\mathcal{P}(\cdot, A)$ is a Poisson random variable with parameter $\mu(A)$:
\[Pr\biggl(\mathcal{P}\left(\cdot, A\right)=k\biggr)=e^{-\mu(A)}\frac{\left(\mu(A)\right)^{k}}{k!},\quad \forall \;k \in \mathbb{N};\]
(iii) for disjoint measurable sets ${A_{1},...,A_{n}}\in \mathcal{G}$, the variables $\mathcal{P}(\cdot, A_{1}),...,\mathcal{P}(\cdot, A_{n})$ are independent.
\end{defn}
To parallel with the Wiener measure, we further define the associated compensated Poisson random measure as follows.
\begin{defn}[\textbf{compensated Poisson random measure}]\label{def:2.2}
Assuming $\mathcal{P}(\cdot, \cdot)$ is a Poisson random measure with the intensity Radon measure $\mu(\cdot)$, then the corresponding compensated Poisson random measure can be constructed by subtracting $\mathcal{P}(\cdot, \cdot)$ by its intensity measure:
\begin{equation*}
\begin{split}\tilde{\mathcal{P}}:
& \Omega \times \mathcal{G} \rightarrow \mathbb{R}\\
& \left(\omega, A\right)\mapsto \tilde{\mathcal{P}}\left(\omega, A\right)=\mathcal{P}(\omega, A)-\mu(A).
\end{split}
\end{equation*}
\end{defn}
From Definition \ref{def:2.2}, one can easily obtain that, for $A \in \mathcal{G}$, $E[\tilde{\mathcal{P}}(\cdot, A)]=0$ and $Var[\tilde{\mathcal{P}}(\cdot, A)]=Var[\mathcal{P}(\cdot, A)]=\mu(A)$. Therefore, we call $\tilde{\mathcal{P}}(\cdot, A)$ a compensated Poisson random variable. Clearly, compensated Poisson random variables parallel to normal random variables with mean zero. To connect Poisson random measure with jump processes, we summarize some results obtained by \cite{cont2003financial} in the following Theorem \ref{thm:jumpprocess}.
\begin{thm}\label{thm:jumpprocess}
Suppose a Poisson random measure $\mathcal{P}( ds, dy)$ on $G=[0,T]\times \mathbb{R}^{d}$ with the intensity measure $\mu(ds\times dy)$ is described as the counting measure associated with a random configuration of points $(T_{n}, Y_{n})\in G$:
\[\mathcal{P}=\sum\limits_{n\geq 1}\delta_{(T_{n}, Y_{n})},\]
where $(T_{n}(\omega), Y_{n}(\omega))\in [0,T]\times \mathbb{R}^{d}$ corresponds to an observation made at time $T_{n}(\omega)$ and described by a random variable $Y_{n}(\omega)$. $f(s,y)$ is a $\mu$-measurable function. Then
\[X(t)\quad=\int_{0}^{t}\int_{\mathbb{R}^{d}}f(s,y)\mathcal{P}( ds, dy),\quad 0 \le t \le T,\]
is a jump process whose jumps happen at the random times $T_{n}$ and have amplitudes given by $f(T_{n},Y_{n})$. Furthermore, the corresponding compensated jump process
\[\tilde{X}(t)\quad=\int_{0}^{t}\int_{\mathbb{R}^{d}}f(s,y)\tilde{\mathcal{P}}( ds, dy),\quad 0 \le t \le T,\] is a martingale adapted to the augmented filtration generated by $\mathcal{P}$.
\end{thm}
Basically, Definitions \ref{def:2.1} and \ref{def:2.2}, Theorem \ref{thm:jumpprocess} summarize the basic properties possessed by Poisson random measure, and characterize the close connection between Poisson random measures and jump processes, hence providing an intuitive understanding towards the construction of the Poisson random measures induced by jump processes. Overall, the compensated Poisson random measure possesses zero-mean and independent increment properties. With these useful tools, we next formally describe the Bermudan option pricing problem under a general jump-diffusion model.
\subsection{Primal problem}
In this paper, we consider a special case of asset price models$-$jump-diffusion processes, i.e., the asset price $X(t)$ satisfies the following stochastic differential equation (SDE):
\begin{equation}\label{eq:2:1}
dX\left( t \right) = b\left( {t,X\left( t \right)} \right)dt + \sigma \left( {t,X\left( t \right)} \right)dW\left( t \right) + \int_{\mathbb{R}^{d}} {J\left( {t,X\left( t \right),y} \right)} \mathcal{P}\left( {dt,dy} \right),
\end{equation}
where $t\in[0,T]$, $X(t) = \left[ X_{1}(t),...,X_{n}(t) \right]$ is a random process with a given initial deterministic value $X(0)=X_{0}\in \mathbb{R}^{n}$, $W(t) = {\left[ W_{1}(t),...,W_{n_{w}}(t)\right]}$ is a standard vector Wiener process, $\mathcal{P}(dt,dy)$ is the Poisson random measure defined on $[0,T]\times\mathbb{R}^{d} \subset {\mathbb{R}^{d+1}} $ with the intensity measure $\mu(dt\times dy)$, the coefficient $b$, $\sigma$ and $J$ are functions $ b:\mathbb{R}\times {\mathbb{R}^n} \to {\mathbb{R}^n}$, $\sigma :\mathbb{R} \times {\mathbb{R}^n} \to {\mathbb{R}^n} \times {\mathbb{R}^{{n_w}}}$ and $J:\mathbb{R} \times {\mathbb{R}^n} \times {\mathbb{R}^d} \to {\mathbb{R}^n}$ satisfying mild continuity conditions (such as uniformly Lipschitz continuous or Holder continuous). Throughout $\mathcal{F}=\{\mathcal{F}_{t}:0\leq t \leq T\}$ is the augmented filtration generated by the Wiener process $W(t)$ and the Poisson random measure $\mathcal{P}$.\\
\\
We consider a Bermudan option based on $X(t)$ that can be exercised at any date from the time set $\Xi  = \{ {T_0},{T_1},...,{T_\mathcal{J} }\}$, with $T_{0}=0$ and $T_{\mathcal{J}}=T$. Given a pricing measure $\mathbb{Q}$ and the filtration $\mathcal{F}$, when exercising at time $T_{j}\in \Xi$, the holder of the option will receive a discounted payoff:
\begin{equation}\label{eq:2:2}
{H_{{T_j}}}\;: = h\left( {{T_j},X\left( {{T_j}} \right)} \right),
\end{equation}
where $h\left( {{T_j},\; \cdot } \right)$ is a Lipschitz continuous function. Our problem is to evaluate the price of the Bermudan option, that is, to find
\begin{equation}\label{eq:2:3}
Primal: \quad\quad    V_{0}^ {\ast}=\sup_{\tau \in \Xi}E\left[h\left(\tau,X\left(\tau \right)\right)|X\left(0\right)=X_{0}\right],
\end{equation}
where $\tau$ is an exercise strategy (in this case, a stopping time adapted to the filtration $\{\mathcal{F}_{T_{j}}:j=0,...,\mathcal{J}\}$) taking values in $\Xi$, and $V_{0}^ {\ast}$ denotes the Bermudan option price at time $T_{0}$ given the initial asset price $X_{0}$.\\
\\
As we stated in the previous section, \cite{longstaff:2001} manage to construct a suboptimal exercise strategy $\tilde{\tau}$ and generate a lower bound $V_{0}^{\tilde{\tau}}$ on the exact option price $V_{0}^ {\ast}$ via a least-squares regression-based approach. We omit the details of their approach and focus on the following dual approach.
\subsection{Review of Dual Approach}

Let $M=\{M_{T_{j}}:j=0,...,\mathcal{J}\}$ with $M_{0}=0$ be a martingale adapted to the filtration $\{\mathcal{F}_{T_{j}}:j=0,...,\mathcal{J} \}$ and $\mathcal{M}$ represents the set of all such martingales. \cite{andersen:2004}, \cite{haugh:2004} show that the dual problem:
\begin{equation}\label{eq:2:4}
Dual:\quad\quad \inf_{M\in \mathcal{M}} \left(E\left[ \max_{0 \le j \le \mathcal{J}}\left( H_{T_{j}} - M_{T_{j}}\right)|X\left(0\right)=X_{0} \right]\right)
\end{equation}
yields the exact option price $V_{0}^ {\ast}$.  Moreover, if we let $M_{T_{j}}$ in \eqref{eq:2:4} be the Doob-Meyer martingale part of the discounted Bermudan price process $V_{T_{j}}^ {\ast}$, denoted by $M_{T_{j}}^{\ast}$, then the infimum in \eqref{eq:2:4} is achieved. Precisely, we have:
\begin{equation}\label{eq:2:5}
 V_{0}^ {\ast} =\quad E\left[ \max_{0 \le j \le \mathcal{J}}\left( H_{T_{j}} - M^{\ast}_{T_{j}}\right)|X\left(0\right)=X_{0} \right].
\end{equation}
In practice, the optimal dual martingale is not accessible to us. Nevertheless, we can still obtain an upper bound with an arbitrary $M \in \mathcal{M}$ via
\begin{equation}\label{eq:2:6}
V_{0}^ {up} \left(M \right)=\quad E\left[ \max_{0 \le j \le \mathcal{J}}\left( H_{T_{j}} - M_{T_{j}}\right)|X\left(0\right)=X_{0} \right].
\end{equation}
It is reasonable to expect that, if $M_{T_{j}}$ is the martingale induced by a good approximation, $V_{T_{j}}$,  of the option price process $V_{T_{j}}^ {\ast}$, then $M_{T_{j}}$ is close to the optimal dual martingale $M_{T_{j}}^{\ast}$ and the resulting upper bound $V_{0}^ {up} \left(M \right)$ should be close to the exact option price $V_{0}^ {\ast}$. Specifically, suppose $V=\{V_{{T_j}}:j=0,...,\mathcal{J}\}$ is some approximation of $V^{\ast}=\{V^{\ast}_{{T_j}}:j=0,...,\mathcal{J}\}$. Consider the following Doob-Meyer decomposition:
\begin{equation}\label{eq:2:7}
{V_{{T_j}}} = {V_0} + {M_{{T_j}}} + {U_{{T_j}}}, \quad j=0,...,\mathcal{J},
\end{equation}
where $V_{0}$ is the approximation of the Bermudan option price at time $T_{0}$ and ${U_{{T_j}}}$ is the residual predictable process. Then we can obtain the martingale component ${M_{{T_j}}}$ in principle via the following recursion:
\begin{equation}\label{eq:2:8}
\left\{\begin{array}{l}
 M_{0}=0,
\\M_{T_{j+1}}=M_{T_{j}}+V_{T_{j+1}}-E_{T_{j}}\left[V_{T_{j+1}}\right].
\end{array}\right.
\end{equation}
where ${E_{{T_j}}}\left[\cdot\right]$ means the conditional expectation is taken with respect to the filtration $\mathcal{F}_{T_{j}}$, i.e., ${E_{{T_j}}}\left[\cdot\right]={E_{{T_j}}}\left[\cdot|\mathcal{F}_{T_{j}}\right]$.\\
\\
\cite{haugh:2004}, \cite{andersen:2004} both use the above theoretical result as the starting point of their algorithms to the upper bounds. The difference between their approaches lies in the ways of generating dual martingales. \cite{haugh:2004} try to first approximate ${V^{\ast}}$ directly by regressing it on certain function basis and then induce the dual martingale by inner simulation, while \cite{andersen:2004} try to first approximate the optimal exercise strategy $\tau^{\ast}$ by a suboptimal exercise strategy $\tilde{\tau}$, then generate the approximation ${V^{\tilde{\tau}}}$ of the option price and the corresponding dual martingale by inner simulation. Nevertheless, due to the nested simulation in approximating the conditional expectation in \eqref{eq:2:8}, both of their algorithms lose the martingale property. Thus the resulting upper bounds are not guaranteed to be true upper bounds. Furthermore, the nested simulation requires huge computational effort. Under limited computational resources, this approach might not be realistic. To overcome these deficiencies, we next develop a new approach to generate a martingale approximation that preserves the martingale property, in a non-nested simulation manner. We expect the resulting upper bounds to be true upper bounds and the computational effort to be significantly reduced.
\section{True Martingale Approach via Non-nested Simulation}
In this section, we will develop an approach that is fundamentally different from previous approaches by \cite{haugh:2004}, \cite{andersen:2004}. By exploiting the special structure of martingales jointly driven by Wiener measure and Poisson random measure, we are able to construct an approximation of $M$ without nested simulation, and thus preserves the martingale property. The following generalized martingale representation theorem provides the intuitive idea in understanding the unique structure of such martingales.
\begin{thm}[\emph{\textbf{Martingale Representation Theorem}}]\label{thm:martingale}
Fix $T>0$. Let $\{W(t):0\leq t\leq T\}$ be a $n_{w}$-dimensional Wiener process and $\mathcal{P}$ be a Poisson random measure on $\left[ {0,T} \right] \times {\mathbb{R}^d}$ with intensity $\mu(dt \times dy)$, independent from $W(t)$. If $M = \{ {M_{T_{j}}:j=0,...,\mathcal{J}}\}$ is a locally square-integrable (real-valued) martingale adapted to the filtration $\{\mathcal{F}_{T_{j}}:j=0,...,\mathcal{J}\}$ with deterministic initial value $M_{0}=0$, then there exist a predictable process $\phi :\;\;\Omega  \times \left[ {0,T} \right] \to \mathbb{R}^{n_{w}} $ and a predictable random function $\psi :\;\;\Omega  \times \left[ {0,T} \right] \times {\mathbb{R}^d} \to \mathbb{R}$ such that
\begin{equation}\label{eq:3:1}
M_{T_{j}} = \int_{0}^{T_{j}} {{\phi _s}dW\left( s \right)}  + \int_0^{T_{j}} {\int_{{\mathbb{R}^d}} {\psi \left( {s,y} \right)} } \tilde{\mathcal{P}}\left( {ds,dy} \right), \quad \quad j=0,...,\mathcal{J},
\end{equation}
where $\tilde{\mathcal{P}}$ is the compensated Poisson random measure induced by $\mathcal{P}$.
\end{thm}
\begin{proof} According to Proposition 9.4 in \cite{cont2003financial}, for the random variable $M_{T}$, there exist a predictable process $\phi :\;\;\Omega  \times \left[ {0,T} \right] \to \mathbb{R}^{n_{w}} $ and a predictable random function $\psi :\;\;\Omega  \times \left[ {0,T} \right] \times {\mathbb{R}^d} \to \mathbb{R}$ such that
\begin{equation*}
M_{T} = \int_{0}^{T} {{\phi _s}dW\left( s \right)}  + \int_0^{T} {\int_{{\mathbb{R}^d}} {\psi \left( {s,y} \right)} } \tilde{\mathcal{P}}\left( {ds,dy} \right), \end{equation*}
where $\tilde{\mathcal{P}}$ is the compensated Poisson random measure induced by $\mathcal{P}$. Given that $M$ is a martingale adapted to the filtration $\{\mathcal{F}_{T_{j}}:j=0,...,\mathcal{J}\}$, and according to the zero-mean and independent increment properties of Wiener measure $W$ and compensated Poisson random measure $\mathcal{\tilde{P}}$, we have
\begin{equation*}
M_{T_{j}} = E\left[M_{T}|\mathcal{F}_{T_{j}}\right]=\int_{0}^{T_{j}} {{\phi _s}dW\left( s \right)}  + \int_0^{T_{j}} {\int_{{\mathbb{R}^d}} {\psi \left( {s,y} \right)} } \tilde{\mathcal{P}}\left( {ds,dy} \right), \quad \mbox{for} \quad j=0,...,\mathcal{J}.
\end{equation*}\qedhere
\end{proof}
Theorem \ref{thm:martingale} can be interpreted as a generalization of the martingale representation theorem for martingales driven by Wiener processes. Indeed, if the intensity $\mu(dt \times dy)$ equals zero, then Theorem \ref{thm:martingale} reduces to the classic martingale representation theorem. Moreover, similar to the Wiener measure $W$, the compensated Poisson random measure $\tilde{\mathcal{P}}$ possesses the zero-mean and independent increment properties, which are essential for the true martingale approximation to preserve the martingale property, as we will elaborate on this point later. Nevertheless, Theorem \ref{thm:martingale} fails to provide an intuitive understanding towards the explicit relationship between $\phi_{t}$, $\psi(t,y)$ and the martingale $M$ itself. We complement this deficiency of Theorem \ref{thm:martingale} by explicitly expressing the integrands as functions of the process that induces the martingale in the following Theorem \ref{thm:derivative}.

\begin{thm}\label{thm:derivative}
Suppose $X_{t}$ follows \eqref{eq:2:1} and induces the jump measure $\mathcal{P}_{X}$, which is a Poisson random measure. Consider the American option price process (continuous-time) $(V_{t})_{0\le t\le T}$ with payoff of the form $h(\cdot,\cdot)$ in \eqref{eq:2:2}. Assuming $V_{t}=v^{a}(t,X(t))$ is a $C^{2}$ function in $X$ and its two partial derivatives are bounded by a constant, then the martingale component of $(V_{t}-V_{0})_{0\le t\le T}$, denoted by $(M_{t})_{0\le t\le T}$ with $M_{0}=0$, is given by:
\begin{equation}\label{eq:3:2}
\small{M_{t}=\int_{0}^{t} \frac{\partial v^{a}\left(s^{-},X_{s^{-}}\right)}{\partial X}\sigma d W_{s} + \int_{0}^{t}\int_{\mathbb{R}^{d}}\left[v^{a}\left(s^{-}, X_{s^{-}}+y\right)-v^{a}\left(s^{-}, X_{s^{-}}\right)\right]\tilde{\mathcal{P}}_{X}\left( {ds,dy} \right), \quad 0 \le t \le T,}
\end{equation}
where $\tilde{\mathcal{P}}_{X}\left( {ds,dy} \right)$ is the compensated Poisson random measure induced by $\mathcal{P}_{X}$.
\end{thm}
\begin{proof}
Apply Proposition 8.16 in \cite{cont2003financial} to jump process $X_{t}$ and we can immediately obtain the result above.\qedhere
\end{proof}
\begin{rem}\label{rem:3:1}
In practice, the asset price process $X_{t}$, which is usually an exponential of a compound Poisson process (see the numerical example \eqref{eq:5:1}), induces a very complicated jump measure $P_{X}$ that can hardly be simulated numerically. To apply Theorem \ref{thm:derivative}, we can introduce an intermediate bridge function $S(\cdot)$ such that $S(t)=S(X(t))$ induces a relatively easy-to-simulate Poisson random measure $P_{S}$ without significantly increasing the complexity of function $v^{a}(\cdot,\cdot)$. In the numerical example we consider later, we will choose a specific function $S(\cdot)$ to achieve this goal.
\end{rem}
Theorem \ref{thm:derivative} implies that $\phi_{t}$ is close to the derivative of the Bermudan option price, while $\psi(t,y)$ is close to the Bermudan option price increment caused by the jump. If we want to approximate these integrands, we should start with the derivative and the increment of certain option price that is close to Bermudan option price, such as the European option price. Specifically, in constructing the T-M algorithm, we will try to use least-squares regression method to approximate the integrands. Therefore we should incorporate the derivative and the increment of the European option price into the function bases for $\phi_{t}$ and $\psi(t,y)$ respectively. Moreover, Remark \ref{rem:3:1} indicates that the choice of the Poisson random measure $\mathcal{P}_{S}$ (or function $S(\cdot)$) is essential to simplify the representation of the martingale. In fact, the choice of $S(\cdot)$ should balance the tradeoff between the complexity of $\mathcal{P}_{S}$ and the complexity of function $v(\cdot,\cdot)$. \\
%
%
\\
Inspired by Theorem \ref{thm:martingale} and Theorem \ref{thm:derivative}, and following \cite{Belomestny:2009}'s work, if one tries to approximate the martingale $M_{T_{j}}$, a natural idea is to first estimate the integrands $\phi_{t}$ and $\psi\left(t,y\right)$ in the following expression:
\begin{equation}\label{eq:3:3}
M_{T_{j}} = \int_0^{{T_j}} {{\phi _t}dW\left( t \right)}  + \int_0^{{T_j}} {\int_{{\mathbb{R}^d}} {\psi \left( {t,y} \right)} } \tilde{\mathcal{ P}}\left( {dt,dy} \right),\;\;j = 0,...,\mathcal{J},
\end{equation}
at a finite number of time and space points. Then an approximation of $M_{T_{j}}$ will be represented via $\phi_{t}$ and $\psi\left(t,y\right)$ using the Ito sum (similar to the Riemann sum).\\
\\
We introduce a partition $\pi=\{t_{l}: l=0,1,...,\mathcal{L} \}$ on $\left[0,T\right]$ such that $t_{0}=0,\; t_{\mathcal{L}}=T$ and $\pi \supset \Xi $, and a partition $\mathcal{A} = \{A_{k}: k=0,1,...,\mathcal{K} \}$ on $\mathbb{R}^{d}$ such that $\left\{ \left[t_{l},t_{l+1}\right]\times{{A_k}}\right\}$ are $\mu$-measurable disjoint subsets and $\bigcup\limits_{k = 1}^\mathcal{K} {{A_k}}  = \mathbb{R}^{d}$. Therefore, $P\left( {\left[ {{t_l},{t_{l + 1}}} \right] \times {A_k}} \right)=\int_{t_{l}}^{t_{l+1}}\int_{A_{k}}\mathcal{P}\left(ds,dy\right)$ is a Poisson random variable (regarded as Poisson increment), and $\tilde{P}\left( {\left[ {{t_l},{t_{l + 1}}} \right] \times {A_k}} \right)=\int_{t_{l}}^{t_{l+1}}\int_{A_{k}}\mathcal{\tilde{P}}\left(ds,dy\right)$ is the corresponding compensated Poisson random variable with intensity $\mu \left( {\left[ {t_{l},t_{l+1}} \right] \times A_{k}} \right)$ (regarded as compensated Poisson increment). We denote the magnitude of partitions $\pi$ and $\mathcal{A}$ as $|\pi|$ and $|\mathcal{A}|$ respectively, i.e., $|\pi|=\max\limits_{0<l\leq \mathcal{L}}\left(t_{l}-t_{l-1}\right)$ and $|\mathcal{A}|=\max\limits_{1\leq k\leq \mathcal{K} }\int_{{A_k}} {f\left( {y} \right)dy}$.\\
\\
From \eqref{eq:2:7}, we have
\begin{equation}\label{eq:3:4}
V_{T_{j+1}}-V_{T_{j}}=\left(M_{T_{j+1}}-M_{T_{j}}\right)+\left(U_{T_{j+1}}-U_{T_{j}}\right),\;\;j = 0,...,\mathcal{J}.
\end{equation}
Combining with the Ito sum of $M_{T_{j+1}}$ in \eqref{eq:3:3}, we have
\begin{equation}\label{eq:3:5}
\begin{array}{l}
{V_{{T_{j + 1}}}} - {V_{{T_j}}} \approx \sum\limits_{{T_j} \le {t_l} < {T_{j + 1}}} {{\phi _{{t_l}}}\left( {W\left( {{t_{l + 1}}} \right) - W\left( {{t_l}} \right)} \right)} \\
\;\;\;\;\;\;\;\;\;\;\;\;\;\;\;\;\;\;\;\;+ \sum\limits_{{T_j} \le {t_l} < {T_{j + 1}}} {\sum\limits_{k = 1}^\mathcal{K} {\psi \left( {{t_l},y_{k}} \right)\tilde{P}\left( {\left[ {{t_l},{t_{l + 1}}} \right] \times {A_k}} \right)} }  + {U_{{T_{j + 1}}}} - {U_{{T_j}}},
\end{array}
\end{equation}
where $y_{k}\in A_{k}$ is a representative value, and we will keep using this notation thereafter. Multiplying both sides of \eqref{eq:3:5} by the Wiener process increment ${\left( {W\left( {{t_{l + 1}}} \right) - W\left( {{t_l}} \right)} \right)}$ and taking conditional expectations  with respect to the filtration $\mathcal{F}_{t_{l}}$, we obtain
\begin{equation}\label{eq:3:6}
{\phi _{{t_l}}} \approx \frac{1}{{{t_{l + 1}} - {t_l}}}{E_{{t_l}}}\left[ {\left( {W\left( {{t_{l + 1}}} \right) - W\left( {{t_l}} \right)} \right){V_{{T_{j + 1}}}}} \right],\;\;\;{T_j} \le {t_l} < {T_{j + 1}},
\end{equation}
where we use the $\mathcal{F}$-predictability of $U$, the independent increment property of $W(t)$ and the independence between $W$ and $\mathcal{P}$.\\
\\
Similarly, if we multiply both sides of \eqref{eq:3:5} by the compensated Poisson random variable ${\tilde{P}\left( {\left[ {{t_l},{t_{l + 1}}} \right] \times {A_k}} \right)}$ and take the conditional expectations with respect to the filtration $\mathcal{F}_{t_{l}}$, we can obtain
\begin{equation}\label{eq:3:7}
\small{\psi \left( {{t_l},y_{k}} \right) \approx \frac{1}{\mu \left( {\left[ {t_{l},t_{l+1}} \right] \times A_{k}} \right)}{E_{{t_l}}}\left[ {\tilde{P}\left( {\left[ {{t_l},{t_{l + 1}}} \right] \times {A_k}} \right){V_{{T_{j + 1}}}}} \right],{T_j} \le {t_l} < {T_{j + 1}} ,1 \le k \le \mathcal{K},}
\end{equation}
where we again use the $\mathcal{F}$-predictability of $U$, the independent increment property of $\mathcal{P}$ and the independence between $\mathcal{P}$ and $W$.\\
\\
Motivated by expressions \eqref{eq:3:6} and \eqref{eq:3:7}, we denote the approximation of ${\phi _{{t_l}}}$ and $\psi \left( {{t_l},y_{k}} \right)$ by $\phi _{{t_l}}^{\pi ,\mathcal{A}}$ and ${\psi ^{\pi ,\mathcal{A}}}\left( {{t_l},y_{k}} \right)$ respectively, which are defined as follows:
\begin{equation}\label{eq:3:8}
\phi _{{t_l}}^{\pi ,\mathcal{A}}\;: = \frac{1}{{\Delta _l^\pi }}{E_{{t_l}}}\left[ {\left( {{\Delta ^\pi }{W_l}} \right){V_{{T_{j + 1}}}}} \right],\;\;\;{T_j} \le {t_l} < {T_{j + 1}},
\end{equation}
and
\begin{equation}\label{eq:3:9}
\small{{\psi ^{\pi ,\mathcal{A}}}\left( {{t_l},y_{k}} \right): = \frac{1}{\mu \left( {\left[ {t_{l},t_{l+1}} \right] \times A_{k}} \right)}{E_{{t_l}}}\left[ {\tilde{P}\left( {\left[ {{t_l},{t_{l + 1}}} \right] \times {A_k}} \right){V_{{T_{j + 1}}}}} \right],{T_j} \le {t_l} < {T_{j + 1}}, 1 \le k \le \mathcal{K},}
\end{equation}
where $\Delta_{l}^{\pi}$ and $\Delta^{\pi}W_{l}$ represent the increments of time $t$ and the Winer process $W(t)$ respectively, i.e. $\Delta_{l}^{\pi}=\left(t_{l+1}-t_{l}\right)$ and $\Delta^{\pi}W_{l}=\left(W_{l+1}-W_{l}\right)$. Therefore, we can construct an approximation of $M_{T_{j}}$, denoted by $M_{{T_j}}^{\pi ,\mathcal{A}}$, as follows:
\begin{equation}\label{eq:3:10}
M_{{T_j}}^{\pi ,\mathcal{A}}\;: = \;\sum\limits_{0 \le {t_l} < {T_{j }}} {\phi _{{t_l}}^{\pi ,\mathcal{A}}\left( {{\Delta ^\pi }{W_l}} \right)}  + \sum\limits_{0 \le {t_l} < {T_{j}}} {\sum\limits_{k = 1}^\mathcal{K} {{\psi ^{\pi ,\mathcal{A}}}\left( {{t_l},y_{k}} \right)\tilde{P}\left( {\left[ {{t_l},{t_{l + 1}}} \right] \times {A_k}} \right)} } .
\end{equation}
The construction procedure of $M_{{T_j}}^{\pi ,\mathcal{A}}$ can be summarized in the following Algorithm 1.
\begin{algorithm}
\caption{\textbf{\emph{\quad \quad Construction of the Martingale Approximation $M^{\pi ,\mathcal{A}}$}}}
\label{alg:1}
Step 1: Express $M_{T_{j}}$ as an integral of $\phi(t)$ and $\psi(t,y)$ via \eqref{eq:3:3}.\\
Step 2: Approximate ${\phi _{{t_l}}}$ by $\phi _{{t_l}}^{\pi ,\mathcal{A}}$ via \eqref{eq:3:8} and $\psi \left( {{t_l},y_{k}} \right)$ by ${\psi ^{\pi ,\mathcal{A}}}\left( {{t_l},y_{k}} \right)$ via \eqref{eq:3:9} respectively.\\
Step 3: Construct the approximation of $M_{T_{j}}$, denoted by $M_{{T_j}}^{\pi ,\mathcal{A}}$, via \eqref{eq:3:10}.
\end{algorithm}
\\
Under the pure-diffusion models, \cite{Belomestny:2009} construct the approximation of $M_{{T_j}}$, denoted by $M_{{T_j}}^{\pi}$ to preserve the martingale property. Here we generalize their techniques to the approximation of the jump part of the martingale under the jump-diffusion models. We observe that, similar to regarding $\phi_{t}$ as a random function of time, we can regard $\psi\left(t,y\right)$ in \eqref{eq:3:5} as a random function of both time and space variables. By properly constructing the Poisson random measure and partitioning the supporting space $\mathbb{R}^{d}$ with respect to the Poisson random measure, we are able to construct the Ito sum to approximate the jump part of $M_{{T_j}}$.\\
\\
Notice that $M^{\pi ,\mathcal{A}}=\{M_{{T_j}}^{\pi ,\mathcal{A}}:j=0,...,\mathcal{J}\}$ remains to be a martingale adapted to the filtration $\{\mathcal{F}_{T_{j}}:j=0,...,\mathcal{J}\}$, based on its structure. We formally state this result in the following theorem.
\begin{thm}[\textbf{\emph{True Martingale}}]\label{thm:3:3} If an approximation of $M$, denoted by $M^{\pi ,\mathcal{A}}$, is constructed using Algorithm \ref{alg:1}, then $M^{\pi ,\mathcal{A}}$ is still a martingale adapted to the filtration $\{\mathcal{F}_{T_{j}}:j=0,...,\mathcal{J}\}$.
\end{thm}

\begin{proof} To show $M^{\pi ,\mathcal{A}}$ is a martingale adapted to the filtration $\{\mathcal{F}_{T_{j}}:j=0,...,\mathcal{J}\}$, it is sufficient to show that for $0\le j_{1} < j_{2} \le \mathcal{J}$,  $E\left[M_{T_{j_{2}}}|\mathcal{F}_{T_{j_{1}}}\right]=M_{T_{j_{1}}}$.\\
\\
For $0\le l\le \mathcal{L}$ and $1\leq k\leq \mathcal{K}$, $\phi _{{t_l}}^{\pi ,\mathcal{A}}$ and ${\psi ^{\pi ,\mathcal{A}}}\left( {{t_l},y_{k}} \right)$ are function of $t_{l}$ and $X_{t_{l}}$. Hence, they are independent from both $\Delta^{\pi}W_{l}$ and $\tilde{P}\left( {\left[ {{t_l},{t_{l + 1}}} \right] \times {A_k}} \right)$. According to the zero-mean property of $\Delta^{\pi}W_{l}$ and $\tilde{P}\left( {\left[ {{t_l},{t_{l + 1}}} \right] \times {A_k}} \right)$, we have
\begin{equation*}
\begin{split}
E\left[M_{T_{j_{2}}}|\mathcal{F}_{T_{j_{1}}}\right]
&=E\left[\sum\limits_{0 \le {t_l} < {T_{j_{2} }}} {\phi _{{t_l}}^{\pi ,\mathcal{A}}\left( {{\Delta ^\pi }{W_l}} \right)}  + \sum\limits_{0 \le {t_l} < {T_{j_{2}}}} {\sum\limits_{k = 1}^\mathcal{K} {{\psi ^{\pi ,\mathcal{A}}}\left( {{t_l},y_{k}} \right)\tilde{P}\left( {\left[ {{t_l},{t_{l + 1}}} \right] \times {A_k}} \right)} }\Big{|}\mathcal{F}_{T_{j_{1}}}\right]\\
&=\sum\limits_{0 \le {t_l} < {T_{j_{1} }}} {\phi _{{t_l}}^{\pi ,\mathcal{A}}\left( {{\Delta ^\pi }{W_l}} \right)}  + \sum\limits_{0 \le {t_l} < {T_{j_{1}}}} {\sum\limits_{k = 1}^\mathcal{K} {{\psi ^{\pi ,\mathcal{A}}}\left( {{t_l},y_{k}} \right)\tilde{P}\left( {\left[ {{t_l},{t_{l + 1}}} \right] \times {A_k}} \right)} }\\
&=M_{T_{j_{1}}}.
\end{split}
\end{equation*}\qedhere
\end{proof}
According to Theorem \ref{thm:3:3}, if we plug $M^{\pi ,\mathcal{A}}$ in \eqref{eq:2:6}, it is easy to see that ${V_{0}^{up}}\left(M^{\pi ,\mathcal{A}}\right)$ is a true upper bound on the Bermudan option price $V_{0}^{\ast}$ in the sense that ${V_{0}^{up}}\left(M^{\pi ,\mathcal{A}}\right)$ is an unbiased expectation for a valid upper bound. Moreover, if we adopt the L-S algorithm to solve the primal problem \eqref{eq:2:3}, we will obtain a suboptimal exercise strategy $\tilde{\tau}$. Exercising $\tilde{\tau}$ along a certain number of sample paths yields an approximation $V_{T_{j}}$ of the Bermudan option price at time $T_{j}$ via $V_{T_{j}}=E_{T_{j}}\left[H_{\tilde{\tau}_{j}}\right]$, where $\tilde{\tau}_{j}$ means the stopping time $\tilde{\tau}$ takes value greater than or equal to $j$. Due to the tower property of conditional expectations, we can rewrite \eqref{eq:3:8} and \eqref{eq:3:9} as
\begin{equation}\label{eq:3:11}
\phi _{{t_l}}^{\pi ,\mathcal{A}}\;: = \frac{1}{{\Delta _l^\pi }}{E_{{t_l}}}\left[ {\left( {{\Delta ^\pi }{W_l}} \right)}H_{\tilde{\tau}_{j+1}} \right],\;\;\;{T_j} \le {t_l} < {T_{j + 1}},
\end{equation}
and
\begin{equation}\label{eq:3:12}
\small{{\psi ^{\pi ,\mathcal{A}}}\left( {{t_l},y_{k}} \right): = \frac{1}{\mu \left( {\left[ {t_{l},t_{l+1}} \right] \times A_{k}} \right)}{E_{{t_l}}}\left[ {\tilde{P}\left( {\left[ {{t_l},{t_{l + 1}}} \right] \times {A_k}} \right)H_{\tilde{\tau}_{j+1}}} \right],{T_j} \le {t_l} < {T_{j + 1}}, 1 \le k \le \mathcal{K}.}
\end{equation}
Through this we avoid the computation of conditional expectations in \eqref{eq:3:8} and \eqref{eq:3:9}, which would incur nested simulation in implementation. Therefore, we can estimate $M_{{T_j}}^{\pi ,\mathcal{A}}$ in \eqref{eq:3:10} via non-nested simulation, and hence significantly improve the computational efficiency.\\
\\
A natural question that arises after we obtain $M^{\pi ,\mathcal{A}}$ is how good it approximates the objective martingale $M$. In the remainder of this section, we will focus on the limiting behaviour of $M^{\pi ,\mathcal{A}}$ as $|\pi|$ goes to zero, and bound the distance between $M^{\pi ,\mathcal{A}}$ and $M$ with $|\pi|$. Precisely, we have the following theorem.
\begin{thm}\label{thm:3:4}
Let ${M_{{T_j}}}$ be the martingale component of ${V_{{T_j}}} = v\left( {T_j},X_{{T_j}} \right)$ and $M_{{T_j}}^{\pi ,\mathcal{A}}$ be its approximation obtained via Algorithm \ref{alg:1}, where $v(T_{j},\cdot)$ are Lipschitz continuous functions. Then there exists a constant $C>0$ such that
\[ E\left[\max_{0\leq j \leq \mathcal{J}}|M_{{T_j}}^{\pi ,\mathcal{A}}-{M_{{T_j}}}|^{2}\right]\leq C|\pi|.\]
\end{thm}
\begin{proof} Fix $T_{j}$. Consider $t_{l}$ such that $T_{j} \leq t_{l} < T_{j+1}$ and $k$ such that $1 \le k \le \mathcal{K}$. According to \eqref{eq:3:8} and \eqref{eq:3:9}, we have:
\begin{equation*}
\begin{split}\phi _{{t_l}}^{\pi ,\mathcal{A}}
&= \frac{1}{{\Delta _l^\pi }}{E_{{t_l}}}\left[ {\left( {{\Delta ^\pi }{W_l}} \right){V_{{T_{j + 1}}}}} \right]\\
&\overset{\left(i\right)}= \frac{1}{{\Delta _l^\pi }}{E_{{t_l}}}\left[ {\left( {{\Delta ^\pi }{W_l}} \right){\left(V_{{T_{j + 1}}}-E_{T_{j}}[V_{T_{j+1}}]\right)}} \right]\\
&\overset{\left(ii\right)}= \frac{1}{\Delta_{l}^{\pi}}E_{t_{l}}\left[\left( {{\Delta ^\pi }{W_l}} \right)\left(M_{T_{j+1}}-M_{T_{j}}\right)\right]\\
&\overset{\left(iii\right)}= \frac{1}{\Delta_{l}^{\pi}}E_{{t_l}}\left[\left(\int_{{t_l}}^{t_{l+1}}d W_{s}\right)\left(\int_{{T_j}}^{T_{j+1}}\phi_{s}d W_{s}+\int_{{T_j}}^{T_{j+1}} {\int_{{\mathbb{R}^d}} {\psi \left( {s,y} \right)} } \tilde{\mathcal{P}}\left( {ds,dy} \right)\right)\right]\\
&\overset{\left(iv\right)}= \frac{1}{\Delta_{l}^{\pi}}E_{{t_l}}\left[\int_{{t_l}}^{t_{l+1}}\phi_{s}ds\right],
\end{split}
\end{equation*}
where the equality (\emph{i}) follows the independent increment property of $W(t)$, equality (\emph{ii}) uses \eqref{eq:2:8}, equality (\emph{iii}) uses \eqref{eq:3:3}, and equality (\emph{iv}) follows the Ito's isometry and the independence between $W$ and $\mathcal{P}$.\\
\\
Similarly, we have
\begin{equation*}
\footnotesize{\begin{split}{ \psi ^{\pi ,\mathcal{A}}}\left( {{t_l},y_{k}} \right)
&=\frac{1}{\mu \left( {\left[ {t_{l},t_{l+1}} \right] \times A_{k}} \right)}E_{t_{l}}\left[ {\tilde{\mathcal{P}}\left( {\left[ {{t_l},{t_{l + 1}}} \right] \times {A_k}} \right){V_{{T_{j + 1}}}}} \right]\\
&\overset{\left(i\right)}= \frac{1}{\mu \left( {\left[ {t_{l},t_{l+1}} \right] \times A_{k}} \right)}{E_{{t_l}}}\left[ {\tilde{\mathcal{P}}\left( {\left[ {{t_l},{t_{l + 1}}} \right] \times {A_k}} \right){\left(V_{{T_{j + 1}}}-E_{T_{j}}[V_{T_{j+1}}]\right)}} \right]\\
&\overset{\left(ii\right)}= \frac{1}{\mu \left( {\left[ {t_{l},t_{l+1}} \right] \times A_{k}} \right)}{E_{{t_l}}}\left[ {\tilde{\mathcal{P}}\left( {\left[ {{t_l},{t_{l + 1}}} \right] \times {A_k}} \right)}\left(M_{T_{j+1}}-M_{T_{j}}\right) \right]\\
&\overset{\left(iii\right)}= \frac{1}{\mu \left( {\left[ {t_{l},t_{l+1}} \right] \times A_{k}} \right)}E_{{t_l}}\biggl[\left(\int_{{t_l}}^{t_{l+1}}\int_{A_{k}}\tilde{\mathcal{P}}\left( {ds,dy} \right)\right)\biggl(\int_{{T_j}}^{T_{j+1}}\phi_{s}d W_{s}
+\int_{{T_j}}^{T_{j+1}} {\int_{{\mathbb{R}^d}} {\psi \left( {s,y} \right)} } \tilde{\mathcal{P}}\left( {ds,dy} \right)\biggr)\biggr]\\
&\overset{\left(iv\right)}= \frac{1}{\mu \left( {\left[ {t_{l},t_{l+1}} \right] \times A_{k}} \right)}E_{{t_l}}\left[\int_{{t_l}}^{t_{l+1}} {\int_{{A_{k}}} {\psi \left( {s,y} \right)} \mu\left( ds\times dy \right) }\right],
\end{split}}
\end{equation*}
where equality (\emph{i}) follows the independent increment property of $\mathcal{\tilde{P}}$, equality (\emph{ii}) uses \eqref{eq:2:8}, equality (\emph{iii}) uses \eqref{eq:3:3}, and equality (\emph{iv}) follows Ito's isometry and the independence between $\mathcal{P}$ and $W$.\\
\\
Furthermore, from \eqref{eq:2:8} and \eqref{eq:3:3}, we have:
\begin{equation}\label{eq:3:13}
\begin{split}{V_{{T_{j + 1}}}} - {E_{{T_j}}}\left[ {{V_{{T_{j + 1}}}}} \right]
&={M_{{T_{j + 1}}}} - {M_{{T_j}}}\\
&=\int_{T_{j}}^{T_{j+1}} {{\phi _s}dW\left( s \right)}  + \int_{T_{j}}^{T_{j+1}} {\int_{{\mathbb{R}^d}} {\psi \left( {s,y} \right)} } \tilde{\mathcal{P}}\left( {ds,dy} \right).
\end{split}
\end{equation}
If we reorganize \eqref{eq:2:1} and \eqref{eq:3:13} together as the following Forward-Backward SDE (FBSDE) on $\left[T_{j},T_{j+1}\right]$
\begin{equation*}
\left\{\begin{array}{l}
X_{t}= X_{T_{j}}+\int_{T_{j}}^{t}b(s,X_{s})ds+\int_{T_{j}}^{t}\sigma(s,X_{s})dW_{s}+\int_{T_{j}}^{t}\int_{\mathbb{R}^{d}}{J\left( {s,X\left( s \right),y} \right)} \tilde{\mathcal{P}}\left( {ds,dy} \right)
\\
Y_{t}:=v\left( T_{j+1},X_{T_{j}}\right)-\int_{t}^{T_{j+1}}\phi_{s}dW_{s}-\int_{t}^{T_{j+1}} {\int_{{\mathbb{R}^d}} {\psi \left( {s,y} \right)} } \tilde{\mathcal{P}}\left( {ds,dy} \right)
\end{array}\right.,
\end{equation*}
then according to the results obtained by \cite{bouchard2008discrete} (see Theorem 2.1 in \cite{bouchard2008discrete}), we have
\begin{equation}\label{eq:3:14}
E\left[\sum\limits_{T_{j} \le {t_l} < {T_{j+1 }}}\int_{{t_l}}^{t_{l+1}}|\phi_{s}-\phi _{{t_l}}^{\pi ,\mathcal{A}}|^{2} ds\right] \leq C_{j}|\pi|,
\end{equation}
and
\begin{equation}\label{eq:3:15}
E\left[\sum\limits_{T_{j} \le {t_l} < {T_{j+1 }}}\sum\limits_{k = 1}^\mathcal{K}
\int_{{t_l}}^{t_{l+1}}\int_{A_{k}}|\psi(s,y)-{\psi ^{\pi ,\mathcal{A}}}\left( {{t_l},y_{k}} \right)|^{2} \mu\left( ds\times dy \right)\right]\leq C_{j}|\pi|,
\end{equation}
for some constant $C_{j}$. Therefore, we have
\begin{equation}\label{eq:3:16}
\begin{split}& \quad \; E\left[\max_{0\leq j \leq \mathcal{J}} \lvert M_{{T_j}}^{\pi ,\mathcal{A}}-{M_{{T_j}}}\rvert^{2}\right] \overset{\left(i\right)}\leq 4E\left[\lvert M_{T}^{\pi ,\mathcal{A}}-M_{T}\rvert^2\right] \\
&= 4E\biggl[\biggl\lvert \sum_{j=0}^{\mathcal{J}-1}\biggl(\sum\limits_{T_{j}\leq t_{l} < T_{j+1}}\int_{t_{l}}^{t_{l+1}} {\left({\phi _t}-\phi_{t_{l}}^{\pi,\mathcal{A}}\right)dW\left( t \right)}\\
&\quad + \sum\limits_{T_{j}\leq t_{l} < T_{j+1}}\sum_{k=1}^{\mathcal{K}}\int_{t_{l}}^{t_{l+1}} {\int_{A_{k}} {\left(\psi \left( {t,y} \right)-\psi^{\pi,\mathcal{A}}\left(t_{l},y_{k}\right)\right)} } \tilde{\mathcal{ P}}\left( {dt,dy} \right) \biggr)\biggr\rvert^{2}\biggr]\\
&\overset{\left(ii\right)}=4\sum \limits_{j=0}^{\mathcal{J}-1}\biggl(E\left[\sum\limits_{T_{j} \leq {t_l} < {T_{j+1}}}\int_{{t_l}}^{t_{l+1}}|\phi_{s}-\phi _{{t_l}}^{\pi ,\mathcal{A}}|^{2} ds\right]\\
&\quad +E\left[\sum\limits_{ T_{j} \leq {t_l} < {T_{j+1 }}}\sum\limits_{k = 1}^{\mathcal{K}}\int_{{t_l}}^{t_{l+1}}\int_{A_{k}}|\psi(s,y)-{\psi ^{\pi ,\mathcal{A}}}\left( {{t_l},y_{k}} \right)|^{2}\mu\left( ds\times dy \right)\right]\biggr)\\
&\leq C|\pi|,
\end{split}
\end{equation}
where the inequality $(i)$ follows Doob's $L^{p}$ maximal inequality, equality $(ii)$ follows Ito's isometry and the independence between $\mathcal{P}$ and $W$, and $C=8\sum\limits_{j=0}^{\mathcal{J}-1}C_{j}$.\qedhere
\end{proof}
According to the relationship between $M$ and $V_{0}^ {up} \left(M \right)$ in \eqref{eq:2:6}, we can immediately obtain the following corollary on the quality of uppers bounds.
\begin{corollary}\label{cor:3:1}
Under the assumptions of Theorem \ref{thm:3:4}, we have
\begin{equation*}
|V_{0}^ {up} (M^{\pi ,\mathcal{A}})-V_{0}^ {up} \left(M \right)|^{2} \leq C |\pi|.
\end{equation*}
\end{corollary}


Corollary \ref{cor:3:1} implies that, if we want to obtain a tight upper bound $V_{0}^ {up} (M^{\pi ,\mathcal{A}})$, we have to partition $\left[0,T\right]$ sufficiently small. Another key procedure to ensure the successful implementation of our algorithm is to find a good way to approximate $\phi^{\pi ,\mathcal{A}}=\{\phi _{{t_l}}^{\pi ,\mathcal{A}}:l=0,...,\mathcal{L}\}$ in \eqref{eq:3:11} and $\psi^{\pi ,\mathcal{A}}=\{\psi^{\pi ,\mathcal{A}}(t_{l},y_{k}):l=0,...,\mathcal{L},k=0,...,\mathcal{K}\}$ in \eqref{eq:3:12} sufficiently well with the least computational effort. In next section, we will approximate $\phi^{\pi ,\mathcal{A}}$ and $\psi^{\pi ,\mathcal{A}}$ via a least-squares regression-based approach and describe the detailed algorithm towards the upper bounds on the Bermudan option price.
\section{True Martingale Algorithm}
We will formally describe the T-M algorithm based on the construction of the martingale approximation $M^{\pi ,\mathcal{A}}$ in section 3. The outline of the T-M algorithm consists of four steps in order: generating a suboptimal exercise strategy $\tilde{\tau}$, approximating the integrands $\phi^{\pi ,\mathcal{A}}$ and $\psi^{\pi ,\mathcal{A}}$, constructing the martingale approximation $M^{\pi ,\mathcal{A}}$, and generating true upper bounds ${V_{0}^{up}}\left(\hat{M}^{\pi,\mathcal{A}} \right)$ on the option price.\\
\\
First, let's start with generating the suboptimal exercise strategy $\tilde{\tau}$. It not only provides the lower bound, but also plays an important role in approximating the integrands $\phi^{\pi ,\mathcal{A}}$ and $\psi^{\pi ,\mathcal{A}}$. We adopt the least-squares regression-based approach proposed by \cite{longstaff:2001} to generate the suboptimal exercise strategy $\tilde{\tau}$ and the corresponding approximation of option price process ${\bar{V}_{{T_j}}}$ at time $T_{j}$, of the form
\begin{equation}\label{eq:4:1}
{\bar{V}_{{T_j}}} = v\left( {{T_j},{X^{\bar{\pi},\bar{\mathcal{A}}}_{{T_j}}}} \right),
\end{equation}
where $\bar{\pi} \supset \pi$, $\bar{\mathcal{A}} \supset \mathcal{A}$ are employed to simulate the discretized asset price process $\{X^{\bar{\pi},\bar{\mathcal{A}}}\}$. \\
\\
Second, let us approximate the integrands $\phi^{\pi ,\mathcal{A}}$ and $\psi^{\pi ,\mathcal{A}}$. To avoid confusion, we denote
\begin{equation*}
\small{\left\{ \begin{array}{l}
\bar{\phi} _{{t_l}}^{\pi ,\mathcal{A}} = \frac{1}{{\Delta _l^\pi }}E_{t_{l}}\left[ {\left( {{\Delta ^\pi }{W_l}} \right)v\left( {{T_{j + 1}},{X^{\bar{\pi},\bar{\mathcal{A}}}_{T_{j+1}}}} \right)} \right],\;\;\;{T_j} \le {t_l} < {T_{j + 1}}\\
{\bar{\psi} ^{\pi ,\mathcal{A}}}\left( {{t_l},y_{k}} \right) = \frac{1}{\mu \left( {\left[ {t_{l},t_{l+1}} \right] \times A_{k}} \right)}E_{t_{l}}\left[ {\tilde{P}\left( {\left[ {{t_l},{t_{l + 1}}} \right] \times {A_k}} \right)v\left( {{T_{j + 1}},{X^{\bar{\pi},\bar{\mathcal{A}}}_{T_{j+1}}}} \right)} \right],{T_j} \le {t_l} < {T_{j + 1}},1 \le k \le \mathcal{K}
\end{array} \right.}
\end{equation*}
as the counterparts of $\phi _{{t_l}}^{\pi ,\mathcal{A}}$ and ${\psi^{\pi ,\mathcal{A}}}\left( {{t_l},y_{k}} \right)$ respectively, under the discretized asset price ${X^{\bar{\pi},\bar{\mathcal{A}}}_{T_{j+1}}}$. Inspired by \cite{longstaff:2001}'s least-squares regression approach to approximating the continuation values, we apply a similar regression technique to approximate ${\bar{\phi} ^\pi }$ and ${\bar{\psi} ^{\pi ,\mathcal{A}}}$. Specifically, the function bases chosen to regress $\bar{\phi} _{{t_l}}^{\pi ,\mathcal{A}}$ and ${\bar{\psi} ^{\pi ,\mathcal{A}}}\left( {{t_l},y_{k}} \right)$ are row function vectors $\rho^{W}\left(t_{l},{X^{\bar{\pi},\bar{\mathcal{A}}}_{{t_l}}}\right)=\left(\rho^{W}_{i}(t_{l},{X^{\bar{\pi},\bar{\mathcal{A}}}_{{t_l}}})\right)_{i=1,...,I_{1}}$  and  $\rho^{\mathcal{P}}\left(t_{l},y_{k},{X^{\bar{\pi},\bar{\mathcal{A}}}_{{t_l}}}\right)=\left(\rho_{i}^{\mathcal{P}}(t_{l},y_{k},{X^{\bar{\pi},\bar{\mathcal{A}}}_{{t_l}}})\right)_{i=1,...,I_{2}}$  respectively, where $I_{1}$ and $I_{2}$ are the dimensions of the function bases. If we simulate $N$ independent samples of Wiener increments ${\Delta ^\pi }{W_l}$, denoted by $\left\{ {\Delta_{n} ^{\pi }{W_l}:\;\;l = 1,...\mathcal{L}, n=1,...,N } \right\}$,
and $N$ independent samples of Poisson increments \\${P}\left( {\left[ {{t_l},{t_{l + 1}}} \right] \times {A_k}} \right)$, denoted by
$\{ {P}_{n}\left( {\left[ {{t_l},{t_{l + 1}}} \right] \times {A_k}} \right):l=1,...,\mathcal{L},k=1,...,\mathcal{K},n=1,...,N \}$,
and based on which we construct the sample paths of the asset price $\{{X^{\bar{\pi},\bar{\mathcal{A}}}_{{t_l},n}}\}_{l=0,...,\mathcal{L},n=1,...,N}$, then we can obtain the regressed coefficients ${{\hat \alpha }_{t_l}} $ and ${{\hat \beta }_{t_{l},k}}$, for $T_{j}\leq t_{l}< T_{j+1}$ and $1\leq k\leq \mathcal{K}$, via
\begin{equation}\label{eq:4:2}
\left\{ \begin{array}{l}
{{\hat \alpha }_{t_l}} \quad= \arg \;\mathop {\min }\limits_{\alpha  \in {\mathbb{R}^{I_{1}}}} \left\{ {\sum\limits_{n = 1}^{N} {{{\left| {\frac{{\Delta _n^\pi {{ W}_l}}}{{\Delta _l^\pi }}H_{\tilde{\tau}_{j+1}}\left({X^{\bar{\pi},\bar{\mathcal{A}}}_{T_{j+1},n}} \right) - \rho^{W}\left(t_{l},{X^{\bar{\pi},\bar{\mathcal{A}}}_{{t_l},n}}\right)\alpha } \right|}^2}} } \right\}\\
{{\hat \beta }_{t_{l},k}} \quad= \arg \;\mathop {\min }\limits_{\beta  \in {\mathbb{R}^{I_{2}}}} \left\{ {\sum\limits_{n = 1}^{N} {{{\left| {\frac{{\tilde{P}}_{n}\left( {\left[ {{t_l},{t_{l + 1}}} \right] \times {A_k}} \right)}{\mu \left( {\left[ {t_{l},t_{l+1}} \right] \times A_{k}} \right)}H_{\tilde{\tau}_{j+1}}\left({X^{\bar{\pi},\bar{\mathcal{A}}}_{T_{j+1},n}} \right) - \rho^{\mathcal{P}}\left(t_{l},y_{k},{X^{\bar{\pi},\bar{\mathcal{A}}}_{{t_l},n}}\right)\beta } \right|}^2}} } \right\}
\end{array} \right.,
\end{equation}
where we employ the tower property to avoid nested simulation, as described in \eqref{eq:3:11} and \eqref{eq:3:12}. Therefore we can compute the approximations of the integrands $\bar{\phi} _{{t_l}}^{\pi,\mathcal{A}}$ and
${\bar{\psi} ^{\pi ,\mathcal{A}}}\left( {t_l},y_{k} \right)$, denoted by ${{\hat \phi }^{\pi ,\mathcal{A}}}\left( {{t_l},x} \right)$ and ${{\hat \psi}^{\pi,\mathcal{A}}}\left( {{t_l},y_{k},x} \right)$ respectively, via
\begin{equation}\label{eq:4:3}
\left\{ \begin{array}{l}
{{\hat \phi }^{\pi ,\mathcal{A}}}\left( {{t_l},x} \right) = \rho^{W} \left( {t_{l},x} \right){{\hat \alpha }_{t_l}}\\
{{\hat \psi }^{\pi ,\mathcal{A}}}\left( {{t_l},y_{k},x} \right) = \rho^{\mathcal{P}} \left( {t_{l},y_{k},x} \right){{\hat \beta }_{t_{l},k}}
\end{array} \right.\;\;\;.
\end{equation}
\\
Next, with fixed $\hat{\alpha}$ and $\hat{\beta}$, we construct an approximation of $M^{\pi,\mathcal{A}}$, denoted by $\hat{M}^{\pi,\mathcal{A}}$, by combining the approximation ${{\hat \phi }^{\pi ,\mathcal{A}}}$ and ${{\hat \psi }^{\pi ,\mathcal{A}}}$ of the integrands with the Euler scheme of system \eqref{eq:2:1}. Precisely, we have
\begin{equation}\label{eq:4:4}
\small{\hat{M}^{\pi,\mathcal{A}}_{{T_j}}:= \sum\limits_{0 \le {t_l} < {T_j}} {{{\hat \phi }^{\pi ,\mathcal{A}}}\left( {{t_l},{X^{\bar{\pi},\bar{\mathcal{A}}}_{{t_l}}}} \right)} \left( {{\Delta ^\pi }{W_l}} \right) + \sum\limits_{0 \le {t_l} < {T_j}} {\sum\limits_{k = 1}^\mathcal{K} {{{\hat \psi }^{\pi ,\mathcal{A}}}\left( {{t_l},y_{k},{X^{\bar{\pi},\bar{\mathcal{A}}}_{{t_l}}}} \right)\tilde{P}\left( {\left[ {{t_l},{t_{l + 1}}} \right] \times {A_k}} \right)} }.}
\end{equation}
Obviously, $\hat{M}^{\pi,\mathcal{A}}$ remains to be a martingale adapted to the filtration $\{\mathcal{F}_{T_{j}}:j=0,...,\mathcal{J}\}$. Consequently, ${V_{0}^{up}}\left(\hat{M}^{\pi,\mathcal{A}}\right)$ is a true upper bound on the Bermudan option price $V_{0}^ {\ast}$. To this end we have finished the construction of the true martingale approximation $\hat{M}^{\pi,\mathcal{A}}$. A natural question is that how good $\hat{M}^{\pi,\mathcal{A}}$ approximates the objective martingale $M_{T_{j}}$ given that the Euler scheme of the asset-price process and the regression \eqref{eq:4:2} both incur bounded errors. We address this question in the following theorem.
\begin{thm}\label{thm:4:1}
Consider $\bar{V}_{T_{j}}=v\left( {{T_j},{X^{\bar{\pi},\bar{\mathcal{A}}}_{{T_j}}}} \right)$, for $j=1,2,...,\mathcal{J}$, where $v(T_{j},\cdot)$ are Lipschitz continuous functions, $X^{\bar{\pi},\bar{\mathcal{A}}}_{{T_j}}$ is the Euler discretization of $X_{T_{j}}$ corresponding to partitions $\bar{\pi} \supset \pi$ and $\bar{\mathcal{A}}\supset\mathcal{A}$. Let $\bar{M}_{{T_j}}$ be the martingale component of $\bar{V}_{T_{j}}$. Assume that for $0 \le l \le \mathcal{L}-1$ and $1\le k \le \mathcal{K}$,
\begin{equation*}
\left\{ \begin{array}{l}
\|{{\hat \phi }^{\pi ,\mathcal{A}}}\left( {{t_l},x} \right)-{{\bar{\phi}}^{\pi ,\mathcal{A}}}\left( {{t_l},x} \right)\|_{1}^{2} \leq \epsilon \\
\|{{\hat \psi }^{\pi ,\mathcal{A}}}\left( {{t_l},y_{k},x} \right)-{{ \bar{\psi} }^{\pi ,\mathcal{A}}}\left( {{t_l},y_{k},x} \right)\|_{1}^{2} \leq \epsilon
\end{array} \right.
\end{equation*}
for some positive number $\epsilon$, then there exists a constant $\bar{C}>0$ such that
\[ E\left[\max_{0\leq j \leq \mathcal{J}}|\hat{M}_{{T_j}}^{\pi ,\mathcal{A}}-{\bar{M}_{{T_j}}}|^{2}\right]\leq \bar{C}\left(|\pi| + \epsilon \right).\]
\end{thm}
\begin{proof} To avoid confusion of notations, we denote
\[\bar{M}_{{T_j}}^{\pi ,\mathcal{A}}\;: = \;\sum\limits_{0 \le {t_l} < {T_{j }}} {\bar{\phi} _{{t_l}}^{\pi ,\mathcal{A}}\left( {{\Delta ^\pi }{W_l}} \right) }  + \sum\limits_{0 \le {t_l} < {T_{j}}} {\sum\limits_{k = 1}^\mathcal{K} {{\bar{\psi} ^{\pi ,\mathcal{A}}}\left( {{t_l},y_{k}} \right)\tilde{P}\left( {\left[ {{t_l},{t_{l + 1}}} \right] \times {A_k}} \right)} } .\]
Then, we have
\begin{equation*}
\begin{split} &E\left[\max_{0\leq j \leq \mathcal{J}} \lvert \hat{M}_{{T_j}}^{\pi ,\mathcal{A}}-{\bar{M}_{{T_j}}}\rvert^{2}\right]\overset{\left(i\right)}\leq 4E\left[\lvert \hat{M}_{T}^{\pi ,\mathcal{A}}-\bar{M}_{T}\rvert^2\right] \\
\overset{\left(ii\right)}\leq &16E\left[\lvert \hat{M}_{T}^{\pi,\mathcal{A}}-\bar{M}_{T}^{\pi,\mathcal{A}}\rvert^2+|\bar{M}_{T}^{\pi,\mathcal{A}}-M_{T}^{\pi,\mathcal{A}}|^2+
|M_{T}^{\pi,\mathcal{A}}-M_{T}|^2+|M_{T}-\bar{M}_{T}|^2\right]\\
=&16[(\ast)+(\ast\ast)+(\ast\ast\ast)+(\ast\ast \ast \ast)],
\end{split}
\end{equation*}
where inequality ($i$) follows Doob's $L^{p}$ maximal inequality and inequality ($ii$) follows Cauchy's inequality. From the assumption, we have
\begin{equation*}
\begin{split}(\ast)&=E\left[\lvert \hat{M}_{T}^{\pi,\mathcal{A}}-\bar{M}_{T}^{\pi,\mathcal{A}}\rvert^2\right] \\
&=\sum \limits_{j=0}^{\mathcal{J}-1}\biggl(E\left[\sum\limits_{T_{j} \leq {t_l} < {T_{j+1}}}\int_{{t_l}}^{t_{l+1}}\left(|\hat{\phi} _{{t_l}}^{\pi ,\mathcal{A}}-\bar{\phi}_{{t_l}}^{\pi ,\mathcal{A}}|^{2}\right) ds\right]\\
&\quad +E\left[\sum\limits_{ T_{j} \leq {t_l} < {T_{j+1 }}}\sum\limits_{k = 1}^{\mathcal{K}}\int_{{t_l}}^{t_{l+1}}\int_{A_{k}}\left(|{\hat{\psi} ^{\pi ,\mathcal{A}}}\left( {{t_l},y_{k}} \right)-{\bar{\psi} ^{\pi ,\mathcal{A}}}\left( {{t_l},y_{k}} \right)|^{2}\right)\mu\left( ds\times dy \right)\right]\biggr)\\
&\leq \left(\mu\left([0,T]\times \mathbb{R}^{d}\right)+T\right)\epsilon.
\end{split}
\end{equation*}
From Theorem \ref{thm:3:4}, we have
\[(\ast\ast\ast)=E\left[|M_{T}^{\pi,\mathcal{A}}-M_{T}|^2\right]\leq C|\pi|.\]
As for term $(\ast\ast\ast\ast)$, we have
\begin{equation*}
\begin{split}(\ast\ast\ast\ast)
&\overset{\left(i\right)}=E\left[\left|\sum\limits_{j=1}^{\mathcal{J}}\biggl(v\left(T_{j},X_{T_{j}}\right)-v\left(T_{j},X_{T_{j}}^{\bar{\pi},\bar{\mathcal{A}}}\right)-
E_{T_{j-1}}\left[v\left(T_{j},X_{T_{j}}\right)-v\left(T_{j},X_{T_{j}}^{\bar{\pi},\bar{\mathcal{A}}}\right)\right]\biggr)\right|^2\right]\\
&\overset{\left(ii\right)}\leq L \sum\limits_{j=1}^{\mathcal{J}}E\left[\left|X_{T_{j}}-X_{T_{j}}^{\bar{\pi},\bar{\mathcal{A}}}\right|^2\right]\overset{\left(iii\right)}\leq \bar{L}|\bar{\pi}|\leq \bar{L}|\pi|,
\end{split}
\end{equation*}
where $L$ and $\bar{L}$ are some constants. Here equality ($i$) follows \eqref{eq:2:8}, inequality ($ii$) follows the Lipschitz continuity of $v(T_{j},\cdot)$ and inequality ($iii$) follows the mild continuity conditions that $b$, $\sigma$ and $J$ satisfy.\\
\\
To this point the term left to estimate is $(\ast\ast)$. Notice that, for ${T_j} \le {t_l} < {T_{j + 1}}$ and $1 \le k \le \mathcal{K}$,
\begin{equation*}
\begin{aligned}
&E_{t_{l}}\left[ \left( {{\Delta ^\pi }{W_l}} \right)\left(v\bigl({T_{j+1}},X_{T_{j+1}}^{\bar{\pi},\bar{\mathcal{A}}}\bigr)-v\left( {T_{j+1}},X_{T_{j+1}} \right)\right)\right]^{2}\frac{1}{{\Delta _l^\pi }} \\
=&E_{t_{l}}\Biggl[ \left( {{\Delta ^\pi }{W_l}} \right)\biggl(E_{t_{l+1}}\left[v\bigl( {T_{j+1}},X_{T_{j+1}}^{\bar{\pi},\bar{\mathcal{A}}}\bigr)-v\left( {T_{j+1}},X_{T_{j+1}} \right)\right] \\
&\quad-E_{t_{l}}\left[v\bigl( {T_{j+1}},X_{T_{j+1}}^{\bar{\pi},\bar{\mathcal{A}}}\bigr)-v\left( {T_{j+1}},X_{T_{j+1}} \right)\right]\biggr)\Biggr]^{2}\frac{1}{{\Delta _l^\pi }} \\
\leq &E_{t_{l}}\biggl[E_{t_{l+1}}\left[v\bigl( {T_{j+1}},X_{T_{j+1}}^{\bar{\pi},\bar{\mathcal{A}}}\bigr)-v\left( {T_{j+1}},X_{T_{j+1}} \right)\right]^2
- E_{t_{l}}\left[v\bigl({T_{j+1}},X_{T_{j+1}}^{\bar{\pi},\bar{\mathcal{A}}}\bigr)-v\left({T_{j+1}},X_{T_{j+1}} \right)\right]^2\biggr],
\end{aligned}
\end{equation*}
and similarly
\begin{equation*}
\begin{aligned}
&E_{t_{l}}\left[  \tilde{P}\left( {\left[ {{t_l},{t_{l + 1}}} \right] \times {A_k}} \right)\left(v\bigl({T_{j+1}},X_{T_{j+1}}^{\bar{\pi},\bar{\mathcal{A}}}\bigr)-v\left( {T_{j+1}},X_{T_{j+1}} \right)\right)\right]^{2}\frac{1}{\mu \left( {\left[ {t_{l},t_{l+1}} \right] \times A_{k}} \right)} \\
\leq &E_{t_{l}}\biggl[E_{t_{l+1}}\left[v\bigl( {T_{j+1}},X_{T_{j+1}}^{\bar{\pi},\bar{\mathcal{A}}}\bigr)-v\left( {T_{j+1}},X_{T_{j+1}} \right)\right]^2
- E_{t_{l}}\left[v\bigl({T_{j+1}},X_{T_{j+1}}^{\bar{\pi},\bar{\mathcal{A}}}\bigr)-v\left({T_{j+1}},X_{T_{j+1}} \right)\right]^2\biggr].
\end{aligned}
\end{equation*}
Hence,
\begin{equation*}
\footnotesize{\begin{split}(\ast\ast)
&=\sum\limits_{j=0}^{\mathcal{J}-1}\sum\limits_{T_{j}\leq t_{l} < T_{j+1}} E\left[E_{t_{l}}\left[\left({{\Delta ^\pi }{W_l}}\right)\left(v\bigl( {T_{j+1}},X_{T_{j+1}}^{\bar{\pi},\bar{\mathcal{A}}}\bigr)-v\left( {T_{j+1}},X_{T_{j+1}} \right)\right)\right]^{2}\frac{1}{{\Delta _l^\pi }}\right]\\
&\quad + \sum\limits_{j=0}^{\mathcal{J}-1}\sum\limits_{T_{j}\leq t_{l} < T_{j+1}}\sum\limits_{k=1}^{\mathcal{K}} E\left[E_{t_{l}}\left[\tilde{P}\left( {\left[ {{t_l},{t_{l + 1}}} \right] \times {A_k}} \right)\left(v\bigl( {T_{j+1}},X_{T_{j+1}}^{\bar{\pi},\bar{\mathcal{A}}}\bigr)-v\left( {T_{j+1}},X_{T_{j+1}} \right)\right)\right]^{2}\frac{1}{\mu \left( {\left[ {t_{l},t_{l+1}} \right] \times A_{k}} \right)}\right]\\
&\leq 2\sum\limits_{j=0}^{\mathcal{J}-1}E\left[ \left| v\bigl( {T_{j+1}},X_{T_{j+1}}^{\bar{\pi},\bar{\mathcal{A}}}\bigr)-v\left( {T_{j+1}},X_{T_{j+1}} \right)\right|^{2}\right]\\
&\leq 2L \sum\limits_{j=1}^{\mathcal{J}}E\left[\left|X_{T_{j}}^{\bar{\pi},\bar{\mathcal{A}}}-X_{T_{j}}\right|^2\right]\leq 2\bar{L}|\bar{\pi}| \leq 2\bar{L}|\pi|.
\end{split}}
\end{equation*}
Therefore,
\begin{equation*}
\begin{split} &E\left[\max_{0\leq j \leq \mathcal{J}} \lvert \hat{M}_{{T_j}}^{\pi ,\mathcal{A}}-{\bar{M}_{{T_j}}}\rvert^{2}\right]
\leq 16[(\ast)+(\ast\ast)+(\ast\ast\ast)+(\ast\ast \ast \ast)]\\
&\leq 16[\left(\mu\left([0,T]\times \mathbb{R}^{d}\right)+T\right)\epsilon+2\bar{L}|\pi|+C|\pi|+2\bar{L}|\pi|]\\
&\leq \bar{C}(|\pi|+\epsilon),
\end{split}
\end{equation*}
where $\bar{C}=16 \max\left(4\bar{L}+C,\mu\left([0,T]\times \mathbb{R}^{d}\right)+T\right)$. \qedhere
\end{proof}
Finally, let's estimate ${V_{0}^{up}}\left(\hat{M}^{\pi,\mathcal{A}}\right)$ via \eqref{eq:2:5} by simulating a new set of $\bar{N}$ independent sample paths $\{X_{n}^{\bar{\pi},\bar{\mathcal{A}}}: n = 1,...,\bar{N} \}$. Precisely, an unbiased estimator for ${V_{0}^{up}}\left(\hat{M}^{\pi,\mathcal{A}} \right)$ is given as follows:
\begin{equation}\label{eq:4:5}
{\hat{V}_{0}^{up}}\left(\hat{M}^{\pi,\mathcal{A}} \right) = \frac{1}{{\bar N}}\sum\limits_{n = 1}^{\bar N} {\mathop {\max }\limits_{0 \le j \le \mathcal{J}} } \left[ {h\left( {{T_j},{X^{\bar{\pi},\bar{\mathcal{A}}}_{{T_j},n}}} \right) - \hat{M}^{\pi,\mathcal{A}}_{{T_j},n}} \right],
\end{equation}
where $\hat{M}^{\pi,\mathcal{A}}_{{T_j},n}$ denotes the realization of $\hat{M}^{\pi,\mathcal{A}}_{{T_j}}$ along the sample path ${X^{\bar{\pi},\bar{\mathcal{A}}}_{{T_j},n}}$. We can formally summarize these steps in the following Algorithm 2.
\begin{algorithm}\label{alg:2}
\caption{\quad \emph{\textbf{True Martingale Algorithm}}}\label{alg:2}
Step 1: Apply the L-S algorithm to generate a suboptimal exercise strategy $\tilde{\tau}$.\\
Step 2: Simulate $N$ independent samples of Wiener increments ${\Delta ^\pi }{W_l}$ and $N$ independent samples of Poisson increments ${P}\left( {\left[ {{t_l},{t_{l + 1}}} \right] \times {A_k}} \right)$, for $l=0,...,\mathcal{L}-1$ and $k=1,...,\mathcal{K}$; construct the sample paths of the asset price $\{{X^{\bar{\pi},\bar{\mathcal{A}}}_{{t_l},n}}\}_{l=0,...,\mathcal{L},n=1,...,N}$.\\
Step 3: Obtain the parameters $\hat{\alpha}=\{\hat{\alpha}_{t_{l}}\}_{l=0,...,\mathcal{L}}$ and $\hat{\beta}=\{\hat{\beta}(t_{l},k)\}_{l=0,...,\mathcal{L},k=0,...,\mathcal{K}}$ via least-squares regression \eqref{eq:4:2} with exercising $\tilde{\tau}$ along the sample paths $\{{X^{\bar{\pi},\bar{\mathcal{A}}}_{{t_l},n}}\}_{l=0,...,\mathcal{L},n=1,...,N}$.\\
Step 4: Simulate a new set of $\bar{N}$ independent sample paths $\{{X^{\bar{\pi},\bar{\mathcal{A}}}_{{t_l},n}}\}_{l=0,...,\mathcal{L},n=1,...,\bar{N}}$; compute $\hat{\phi}^{\pi,\mathcal{A}}$ and $\hat{\psi}^{\pi, \mathcal{A}} $ via \eqref{eq:4:3}; construct the martingale approximation $\hat{M}^{\pi,\mathcal{A}}$ via \eqref{eq:4:4}; obtain an unbiased estimator ${\hat{V}_{0}^{up}}(\hat{M}^{\pi,\mathcal{A}})$ for the true upper bound on the Bermudan option price $V_{0}^{\ast}$ via \eqref{eq:4:5}.
\end{algorithm}
\section{Numerical Experiments}
In this section, we will conduct numerical experiments to illustrate the computational efficiency of our proposed T-M algorithm on a Bermudan option pricing problem under a jump-diffusion model. The exact model we consider here falls into the class of jump-diffusion models (see \cite{merton1976option} and \cite{kou2002jump}) reviewed in section 1. Specifically, the asset prices evolve as follows:
\begin{equation}\label{eq:5:1}
\frac{dX\left( t \right)}{X\left(t^{-}\right)} = \left( r-\delta \right)dt + \sigma dW\left( t \right) + d\left( \sum\limits_{i=1}^{P(t)}\left(V_{i}-1\right)\right),
\end{equation}
where $r$ is the constant discount factor, $\delta$ is the constant dividend, $\sigma$ is the constant volatility, $X(t)= \left[X_{1}(t),...,X_{n}(t)\right]$ represents the asset price with a given initial price $X_{0}$,  $W(t) =  \left[W_{1}(t),...,W_{n}(t)\right]$ is a Wiener process, $P(t)$ is a Poisson process with intensity $\lambda$, and $\{V_{i}\}$ is a sequence of independent identically distributed (i.i.d.) nonnegative random variables such that $J=\log(V)$ is the jump amplitude with p.d.f. $f(y)$.  Here $J$ can follow a normal distribution (see \cite{merton1976option}) or a double-exponential distribution (see \cite{kou2008jump}) or various other reasonable distributions. For simplicity, we assume $J$ follows a one-dimensional ($d=1$) normal distribution $N(m,\theta^2)$. We also assume $W(t)$, $P(t)$ and $J$ are mutually independent. \\
\\
To connect dynamics \eqref{eq:5:1} with the jump-diffusion model \eqref{eq:2:1} we have mainly focused on, we should construct a function of $X(t)$, denoted by $S(t)=S(X(t))$ such that dynamics \eqref{eq:5:1} can be easily transformed to an equivalent dynamics jointly driven by the Wiener measure and a Poisson random measure $\mathcal{P}_{S}$. The following proposition provides an intuitive criterion in selecting such a function by explicitly defining the intensity function  $\mu \left( {dt \times dy} \right)$ for the unique Poisson random measure induced by a compound Poisson process.
\begin{prop}[\emph{Proposition 3.5 in \cite{cont2003financial}}]\label{thm:5:1}
Let $S(t)_{t>0}$ be a compound Poisson process with intensity $\lambda$ and jump size distribution $f$. Then the Poisson random
measure $\mathcal{P}_{S}$ induced by $S(t)_{t>0}$ on $\left[ {0,\infty} \right] \times {\mathbb{R}^d}$ has intensity measure $\mu \left( {dt \times dy} \right) = \lambda f\left( {y} \right)dydt$.
\end{prop}
According to Proposition \ref{thm:5:1}, for a compound Poisson process $S(t)$, the compensated Poisson random measure $\mathcal{\tilde{P}}_{S}$ induced by $S(t)$ can be simulated by $\mathcal{\tilde{P}}_{S}=\mathcal{P}_{S}-\lambda f\left( {y} \right)dydt$. Although $X(t)$ given by \eqref{eq:5:2} is not a compound process, $S(t)=\log(X(t))$ is a compound Poisson process, and thus its Poisson random measure $\mathcal{P}_{S}(t,y)$ can be easily simulated according to Proposition \ref{thm:5:1}. Moreover, the simplicity of function $S(\cdot)$ guarantees that the function $g$ in Theorem \ref{thm:derivative} can be easily determined, which is essential for the construction of bases $\rho^{W}$ and $\rho^{\mathcal{P}}$, as we will elaborate on this point later. Now if we incorporate $\mathcal{P}_{S}$ into the asset-price dynamics \eqref{eq:5:1}, we can obtain an equivalent dynamics as
\begin{equation}\label{eq:5:2}
\frac{dX\left( t \right)}{X\left(t^{-}\right)} = \left( r-\delta \right)dt + \sigma dW\left( t \right) + \int_{\mathbb{R}^{d}}y\mathcal{P}_{S}\left(t,y\right).
\end{equation}
\\
Unfortunately, the solution to asset dynamics \eqref{eq:5:1} or \eqref{eq:5:2} is not uniquely determined in the risk-neutral sense, caused by the incompleteness of the market under the jump-diffusion setting. However, we can construct different pricing measures $\mathbb{Q }'s \sim \mathbb{P}$ such that the discounted price $\hat{X}(t)$ is a martingale under $\mathbb{Q}'s$ (c.f. Chapter 10 in \cite{cont2003financial}). Here we will adopt the construction method proposed by \cite{merton1976option}. That is, changing the drift of the Wiener process but leaving other components of \eqref{eq:5:1} unchanged to offset the jump results in a risk-neutral measure $\mathbb{Q}_{M}$, which is a generalization of the unique risk-neutral measure under the Black-Scholes model. Therefore, the solution under $\mathbb{Q}_{M}$ can be easily derived and efficiently simulated. Precisely, the solution to the asset-price dynamics \eqref{eq:5:1} is given by:
\begin{equation}\label{eq:5:3}
X\left(t\right)=X_{0}\exp\biggl[\mu^{M}t+\sigma W^{M}\left(t\right)+\sum\limits^{P\left(t\right)}_{i=1}J_{i}\biggr],\quad t>0,
\end{equation}
where $\mu^{M}=r-\delta-\frac{1}{2}\sigma^2-E\left[e^{J_{i}}-1\right]$ is the new drift, $W^{M}(t)$ is a standard vector Wiener process and $J_{i}'s$ are the i.i.d. random variables according to $J$.\\
\\
Given the equivalence of \eqref{eq:5:1} and \eqref{eq:5:3}, we can simulate $X(t)$ under the risk-neutral measure $\mathbb{Q}_{M}$ by simulating $S(t)$ and its associated Poisson random measure $\mathcal{P}_{S}$. Specifically, we perform the Euler scheme on an equidistant partition $\bar{\pi}$ with $|\bar{\pi}|=0.01$ and a continuously equi-probabilistic partition on $\mathbb{R}^{d}$ with $|\bar{\mathcal{A}}|=0.1$ to simulate the Wiener increments $\{W_{t_{l}}\}$, the Poisson random measure increments $P([t_{l},t_{l+1}]\times A_{k})$, and the resulting sample paths of $X(t)=\exp(S(t))$ according to \eqref{eq:5:3}.\\
\\
We consider a Bermudan Min-Puts on the $n$ assets with price vector $\{X_{1}\left(t\right),...,X_{n}\left(t\right)\}$. In particular, at any time $t \in \Xi  = \{ {T_0},{T_1},...,{T_\mathcal{J} }\}$, the option holder has the right to exercise the option to receive the payoff
\[ h\left(X\left(t\right)\right)=\left(SK-\min\left(X_{1}\left(t\right),...,X_{n}\left(t\right)\right)\right)^{+} ,\]
where $SK$ represents the strike price. The maturity time of the option is $T=1$ and can be exercised at 11 equally-spaced time points, i.e., $T_{j}=j\times T/10, \; j=0,...,10$. Our objective is to solve the Bermudan option pricing problem by providing a lower bound and an upper bound on the exact option price.
\subsection{Suboptimal Exercise Strategies and Lower Bounds}
First, let's adopt the L-S algorithm to generate a suboptimal exercise strategy $\tilde{\tau}$ by regressing the continuation values, and compute the corresponding benchmark lower bound. It turns out this algorithm will be very effective if one can construct good function bases for regression in the sense that the function bases should capture sufficient features of the continuation values. In particular, \cite{andersen:2004} propose a function basis consisting of all monomials of underlying asset prices with degrees less than or equal to three, the European min-put option with maturity T, its square and its cube. Numerical results show that this function basis works extremely well for Bermudan option pricing problems under the pure-diffusion models. The reason is that European option price under the pure-diffusion models can be fast computed via its closed-form, and capture sufficient features of the Bermudan option.\\
\\
For the Bermudan option pricing problem under the jump-diffusion model \eqref{eq:5:1}, the corresponding European option still has a closed-form expression. Specifically, under jump-diffusion model \eqref{eq:5:1}, the explicit form of the European option on $X$ at time $t$ with maturity $T$, denoted by $C^{M}(t,X;T)$, is given by:
\begin{equation}\label{eq:5:4}
C^{M}(t,X;T)= E_{\mathbb{Q}_{M}}\left[h(X_{T})|\mathcal{F}_{t}\right]= \sum_{k \geq 0}\frac{e^{-\lambda(T-t)}\left(\lambda(T-t)\right)^{k}}{k!}C_{\sigma_{k}}^{BS}\left(0,X_{k};T-t\right),
\end{equation}
where $\sigma_{k}^{2}=\sigma^{2}+k\theta^{2}/\left(T-t\right)$,
$X_{k}=X\exp\left(k(m+\frac{\theta^{2}}{2})-\lambda[\exp(m+\frac{\theta^{2}}{2})-1](T-t)\right)$
and
\begin{equation}\label{eq:5:5}
\begin{split}C_{\sigma}^{BS}(0,X;\tau)
&=E\left[H(Xe^{(r-\delta-\frac{\sigma^{2}}{2})\tau+ \sigma W_{\tau}})\right]\\
&=-\sum_{l=1}^{n}X^{l}\frac{e^{-\delta\tau}}{\sqrt{2\pi}}\int_{(-\infty,d^{l}_{-}]}\exp\left(-\frac{z^2}{2}\right)\prod_{l^{\prime}=1,l^{\prime}\neq l}^{n}
\mathcal{N}\left(\frac{\ln\frac{X^{l^{\prime}}}{X^{l}}}{\sigma \sqrt{\tau}}-z-\sigma \sqrt{\tau}\right)dz \\
&\quad + e^{-r\tau} \cdot SK \cdot \left(1-\prod_{l=1}^{n}\left(1-\mathcal{N}\left(d^{l}_{+}\right)\right)\right),\\
\end{split}
\end{equation}
with
\begin{equation}\nonumber
d^{l}_{+}=\frac{\ln\frac{SK}{X^{l}}-\left(r-\delta-\frac{\sigma^{2}}{2}\right)\tau}{\sigma \sqrt{\tau}},\quad\quad d^{l}_{-}=d^{l}_{+}-\sigma \sqrt{\tau},
\end{equation}
and $\mathcal{N}$ denoting the c.d.f. of standard normal distribution, $C_{\sigma}^{BS}(0,X;\tau)$ denoting the European option under the Black-Scholes model with maturity $\tau$, volatility $\sigma$ and initial price $X$. Unfortunately, it is extremely difficult to compute $C^{M}(t,X;T)$ in \eqref{eq:5:4} exactly because of the infinite sum in \eqref{eq:5:4}. A natural modification is to approximate it numerically by finite truncation the summation in \eqref{eq:5:4} and some approximation of the integral in \eqref{eq:5:5}. Naturally, we try to approximate it directly by an European option under a closely-related pure-diffusion model. Surprisingly, this most intuitive one, i.e., the European option under the pure-diffusion model derived simply by removing the jump part of \eqref{eq:5:1} works extremely well in our numerical experiments. To avoid confusion, we refer to it as ``non-jump European option''.  \\
\\
Now the function basis we choose includes all monomials of underlying asset prices with degrees less than or equal to three, the non-jump European option with maturity T, its square and its cube. With this basis, we implement the least-squares regression algorithm, and obtain suboptimal exercise strategies $\tilde{\tau}'s$ and the corresponding lower bounds, as shown in Table \ref{table:5:1}.
\subsection{Benchmark Upper Bounds}
After obtaining the suboptimal exercise strategies $\tilde{\tau}'s$, we adopt the A-B algorithm with nested Monte Carlo simulation to compute the benchmark upper bounds. We report the numerical results in Table \ref{table:5:1}, in which we can see that the A-B algorithm yields extremely tight upper bounds with small duality gaps. This observation indicates two facts. First, the suboptimal exercise strategies $\tilde{\tau}'s$ constructed with the new function basis are nearly optimal, which are crucial for the successful implementation of the T-M algorithm because we also need $\tilde{\tau}'s$ to estimate the regression coefficients $\hat{\alpha}$ and $\hat{\beta}$. Moreover, we will construct the function bases $\rho^{W}$ and $\rho^{\mathcal{P}}$ using the new function basis as a starting point, therefore the effectiveness of the new function basis is a positive indicator for the sufficiency of bases $\rho^{W}$ and $\rho^{\mathcal{P}}$ in terms of capturing features. Second, the A-B algorithm is still very effective under the jump-diffusion models, regardless of the considerable computational effort caused by the nested simulation.
\subsection{Upper Bound by True Martingale Approach}
Next, we will address the computational inefficiency suffered by the A-B algorithm by implementing our proposed T-M algorithm (Algorithm \ref{alg:2}) described in section 4. Notice that we have addressed almost all the implementation details except the choices of the partitions $\pi$ and $\mathcal{A}$, and the bases $\rho^{W}$ and $\rho^{\mathcal{P}}$.\\
\\
First of all, the choice of the partition $\pi$ is essential to balance the tradeoff between the quality of the true martingale approximation and the computational efficiency. It has to be sufficiently small to reduce the overall mean square error between the true martingale approximation and the objective martingale, but not too small so that the computational effort for obtaining martingale approximation $M^{\pi,\mathcal{A}}$ is much less than the computational effort for obtaining the inner sample paths in the A-B algorithm. In fact, a good way to achieve this tradeoff is to perform the regression on a rough partition in the beginning, and then interpolate them piece-wisely constant to a finer partition. To maximize the effect of this two-layer regression-interpolation technique, we choose to perform the regression procedure on the original exercisable dates $\Xi  = \{ {T_0},{T_1},...,{T_\mathcal{J}-1 }\}$ and interpolate the regression coefficients piece-wisely constant to the partition $\bar{\pi}$ of the Euler scheme. Secondly, the choice of the partition $\mathcal{A}$ is less restrictive than the choice of $\pi$ since $|\pi|$ will dominate the error between the martingale approximation and the original martingale (see Theorem\eqref{thm:3:4}) regardless of the choice of $\mathcal{A}$. For the sake of convenience, we let $\mathcal{A}=\mathcal{\bar{A}}$. Therefore the compensated Poisson increments $\{\tilde{P}\left( {\left[ {{t_l},{t_{l + 1}}} \right] \times {A_k}} \right)\}$ in \eqref{eq:4:2} are obtained immediately from the simulation of $X^{\bar{\pi},\mathcal{\bar{A}}}$, and $\mu \left( {\left[ {t_{l},t_{l+1}} \right] \times A_{k}} \right)$ in \eqref{eq:4:2} equals $\lambda\times0.01\times 0.1$ (see Proposition \ref{thm:5:1}). Specifically, we obtain $\{{{\hat \alpha }_{{T_j}}}, j=0,...,\mathcal{J}-1\}$ and $\{{{\hat \beta }_{{T_j},k}},j=0,...,\mathcal{J}-1, k=1,...,\mathcal{K}\}$ via the regression \eqref{eq:4:2}, and set $\hat{\alpha}_{t_{l}}={{\hat \alpha }_{{T_j}}}$ for $t_{l} \in [T_{j},T_{j+1})$ and $\hat{\beta}_{t_{l},k}={{\hat \beta }_{{T_j},k}}$ for $ t_{l} \in [T_{j},T_{j+1}), k=1,...,\mathcal{K}$.\\
\\
Finally, the choice of the bases $\rho^{W}$ and $\rho^{\mathcal{P}}$ affects the accuracy of the martingale approximation $\hat{M}^{\pi,\mathcal{A}}$. According to Theorem \ref{thm:derivative}, the bases $\rho^{W}$ should capture sufficient features of the derivative of the Bermudan option price, while the bases $\rho^{\mathcal{P}}$ should capture sufficient features of the increment of the Bermudan option price caused by the jump. As we showed earlier, the non-jump European option is a good basis function for the Bermudan option price. Therefore, if we apply Ito's lemma on the non-jump European option, we expect the resulting integrands to be good starting points for the bases $\rho^{W}$ and $\rho^{\mathcal{P}}$. Precisely, for $0\leq j\leq \mathcal{J}$, we have
\begin{equation}\label{eq:5:6}
\begin{split}C^{M}\left(t,X_{t};T_{j}\right)
\quad :&=E_{\mathbb{Q}_{M}}\left[h(X_{T_{j}})|\mathcal{F}_{t}\right]\\
&= h(X_{T_{j}})-\int^{T_{j}}_{t}\frac{\partial C^{M}\left(u,X_{u^{-}};T_{j}\right)}{\partial X}X_{u^{-}}\sigma d W_{u}^{M}\\
&\quad -\int^{T_{j}}_{t}\int_{\mathbb{R}^{d}}\left[C^{M}\left(u,X_{u^{-}}\cdot e^{y};T_{j}\right)-C^{M}\left(u,X_{u^{-}};T_{j}\right)\right]\tilde{\mathcal{P}}_{S}
\left(du, dy\right).
\end{split}
\end{equation}
Therefore, the function basis $\rho^{W}\left(t,X_{t^{-}}\right)$ should include the derivative of $C^{BS}\left(t,X_{t^{-}};\cdot\right)$, while the function basis $\rho^{\mathcal{P}}(t,y,X_{t^{-}})$ should include $C^{BS}\left(t,X_{t^{-}}\cdot e^{y};\cdot\right)-C^{BS}\left(t,X_{t^{-}};\cdot\right)$. After simple numerical tests, we find out (see Basis 4 in Table \ref{table:5:3}), for $t\in[T_{j},T_{j+1})$ and $1\leq k\leq \mathcal{K}$, $\rho^{W}\left(t,X_{t^{-}}\right)$ consisting of 1, $\frac{\partial C^{BS}\left(t,X_{t^{-}};T_{j+1}\right)}{\partial X}X_{t^{-}}$ and $\frac{\partial C^{BS}\left(t,X_{t^{-}};T\right)}{\partial X}X_{t^{-}}$, $\rho^{\mathcal{P}}(t,y_{k}, X_{t^{-}})$ consisting of 1, $C^{BS}\left(t,X_{t^{-}}\cdot e^{y_{k}};T_{j+1}\right)-C^{BS}\left(t,X_{t^{-}};T_{j+1}\right)$ and $C^{BS}\left(t,X_{t^{-}}\cdot e^{y_{k}};T\right)-C^{BS}\left(t,X_{t^{-}};T\right)$ yield the tightest upper bounds, where $y_{k}\in A_{k}$ is a representative value.\\
\\
We report the numerical results on the lower bounds by the L-S algorithm, the benchmark upper bounds by the A-B algorithm and the true upper bounds by the T-M algorithm in Table \ref{table:5:1}. The small gaps between the lower bounds and the true upper bounds indicate that the T-M algorithm is quite effective in terms of generating tight true upper bounds. The small length of the confidence intervals of the true upper bounds indicates that T-M algorithm generates good approximations of the optimal dual martingales. The CPU time ratios indicate that T-M algorithm achieves a much higher numerical efficiency compared with the A-B algorithm.
\begin{table}[h]\label{table:5:1}
\centering
\begin{threeparttable}
\caption{Bounds (with 95$\%$ confidence intervals) for Bermudan Min-put options}\label{table:5:1}
\begin{tabular}{ccccccc}
\hline\hline
    &           &            &Lower Bound          &Upper Bound        &Benchmark U-B      &\footnotesize{CPU Time Ratio} \\
$n$ &$\lambda$  &$X_{0}$     &(L-S algorithm)      &(T-M algorithm)    &(A-B algorithm)    &(T-M vs A-B)   \\ \hline
\\[0.5ex]
$1$ &$1$        &$36$        &$5.842\pm0.031$      &$5.970\pm0.031$    &$5.899\pm0.038$    &$\approx$ 1:400   \\
$1$ &$1$        &$40$        &$3.791\pm0.028$      &$3.910\pm0.033$    &$3.856\pm0.036$    &$\approx$ 1:400   \\
$1$ &$1$        &$44$        &$2.383\pm0.024$      &$2.443\pm0.028$    &$2.417\pm0.033$    &$\approx$ 1:400   \\
\\
$1$ &$3$        &$36$        &$7.702\pm0.043$      &$7.899\pm0.030$    &$7.810\pm0.053$    &$\approx$ 1:400   \\
$1$ &$3$        &$40$        &$5.817\pm0.039$      &$5.996\pm0.047$    &$5.894\pm0.050$    &$\approx$ 1:400   \\
$1$ &$3$        &$44$        &$4.352\pm0.036$      &$4.480\pm0.044$    &$4.440\pm0.040$    &$\approx$ 1:400   \\
\\
$2$ &$1$        &$36$        &$8.133\pm0.033$      &$8.308\pm0.045$    &$8.243\pm0.040$    &$\approx$ 1:9   \\
$2$ &$1$        &$40$        &$5.691\pm0.034$      &$5.785\pm0.040$    &$5.755\pm0.043$    &$\approx$ 1:9   \\
$2$ &$1$        &$44$        &$3.765\pm0.028$      &$3.842\pm0.036$    &$3.804\pm0.038$    &$\approx$ 1:9   \\
\\
$2$ &$3$        &$36$        &$9.786\pm0.045$      &$10.038\pm0.061$   &$9.989\pm0.057$    &$\approx$ 1:9   \\
$2$ &$3$        &$40$        &$7.680\pm0.043$      &$7.900\pm0.060$    &$7.845\pm0.057$    &$\approx$ 1:9   \\
$2$ &$3$        &$44$        &$5.941\pm0.040$      &$6.118\pm0.058$    &$6.065\pm0.058$    &$\approx$ 1:9   \\
\hline
\end{tabular}
\begin{tablenotes}
\footnotesize
\item \quad The payoff of the min-put option is:
$(SK-min(X_{1}(t),...,X_{n}(t)))^{+}$. The parameters are: $SK=40, r=4\%,\delta=0,\sigma=20\%,m=6\%,\theta=20\%,T=1,\mathcal{J}=10.$ The jump intensity $\lambda$ is 1 or 3 and the initial price is $X_{0}=(X,...,X)$ with $X=$36, 40 or 44, as shown in the table. We use $N=5\times 10^{4}$ sample paths to estimate the regression coefficients to determine the suboptimal exercise strategy, and we use $N=5\times 10^{4}$ sample paths to estimate the coefficients $\hat{\alpha}$ and $\hat{\beta}$. We use $N_{1}=10^{5}$ sample paths to determine the lower bounds. For the implementation of the A-B algorithm, we use $N_{2}= 10^{3}$ outer sample paths and $N_{3}=5\times 10^{2}$ inner sample paths to determine the benchmark upper bounds and the confidence intervals of appropriate length. For the implementation of the T-M algorithm, we use $\bar{N}=2.5\times 10^{3}$ sample paths to determine the true upper bounds and the confidence intervals of appropriate length.
\end{tablenotes}
\end{threeparttable}
\end{table}
\\
\\
It is instructive to theoretically compare the computational complexity of the T-M algorithm and the A-B algorithm, since the CPU time ratios in Table \ref{table:5:1} are quite different for 1-dimensional problems and 2-dimensional problems. We can use $EE+ES$ to represent the total CPU time, where $EE$ represents the CPU time for evaluating the European option prices and $ES$ represents the CPU time for exercising the strategy along all the sample paths. Simple numerical tests show that when $n=1$, the term $ES$ will dominate the total CPU time because the European option price can be evaluated extremely fast, since the integral in \eqref{eq:5:5} reduces to the c.d.f. of a standard normal distribution. Therefore, the CPU time ratio will be in the order of number ratio of sample paths, which is consistent with the result ($\approx$ 1:400). However, when $n \ge 2$, the term $EE$ will be dominant over $ES$ because the evaluation of the European option price consumes over 95$\%$ of the total CPU time, caused by the evaluation of the integral in \eqref{eq:5:5}. Hence, we should compare the total evaluation times of the European option price for both algorithms to estimate the CPU time ratio. For the A-B algorithm, the total evaluation times is in the order of $(N_{2}\times N_{3} \times \mathcal{J} \times \mathcal{J})$, which is quadratic in the number of exercisable periods $\mathcal{J}$ and will be significantly amplified by the two-layer simulation. For the T-M algorithm, the total evaluation times is in the order of $(\bar{N}\times \mathcal{L}\times\mathcal{K})$, which is linear in the number of exercisable periods, since $\mathcal{L}$ is a linear function of $\mathcal{J}$ (in our case $\mathcal{L}=10\mathcal{J}$) and $\mathcal{K}$ is usually quite small (in our case $\mathcal{K}=10$). Therefore T-M algorithm can achieve a higher order of computational efficiency, which is verified by the numerical results. We can expect that the CPU time ratios (T-M algorithm versus A-B algorithm) to further decrease when the number of exercisable periods increases, and remain stable if we increase the dimension of the problem.\\
\\
An interesting experiment has been conducted to exhibit the quality differences between the upper bounds generated by different bases $\rho^{W}$ and $\rho^{\mathcal{P}}$, and we report the results in Table \ref{table:5:2}. Specifically, the upper bounds of different levels of quality are generated using various bases that are presented in Table \ref{table:5:3}. To summarize, the simplest basis, i.e., $\{1\}$ (Bases 1), results in very poor upper bounds; the standard basis one can come up with, i.e., the polynomials (Bases 2), improves the upper bounds significantly. But the gap is still too large; however, the upper bounds get almost no improvement after we use the non-jump European options as the basis (Bases 3), which indicates that the non-jump European option does not further provide useful features. Finally, the basis $\rho^{W}$ consisting of the deltas of the non-jump European options and the basis $\rho^{\mathcal{P}}$ consisting of the non-jump European option increments (Bases 4) yield desirable upper bounds. These results verify the theoretical analysis about the structure of the optimal dual martingales, as shown in Theorem \ref{thm:derivative}.
\begin{table}[htb]
\centering
\caption{Upper Bounds for different function bases}\label{table:5:2}
\begin{tabular}{cccccc}
\hline\hline
$\lambda$  &$X_{0}$     &Bases 1              &Bases 2            &Bases 3           &Bases 4            \\ \hline
\\
$1$        &$36$        &$6.730\pm0.069$      &$6.283\pm0.042$    &$6.228\pm0.048$   &$5.970\pm0.031$    \\
$1$        &$40$        &$4.789\pm0.074$      &$4.228\pm0.039$    &$4.127\pm0.047$   &$3.910\pm0.033$    \\
$1$        &$44$        &$3.344\pm0.073$      &$2.734\pm0.038$    &$2.665\pm0.044$   &$2.443\pm0.028$    \\
\\
$3$        &$36$        &$8.829\pm0.091$      &$8.338\pm0.059$    &$8.167\pm0.062$   &$7.899\pm0.030$    \\
$3$        &$40$        &$7.086\pm0.101$      &$6.377\pm0.060$    &$6.277\pm0.067$   &$5.996\pm0.047$    \\
$3$        &$44$        &$5.681\pm0.100$      &$4.953\pm0.057$    &$4.752\pm0.061$   &$4.480\pm0.044$    \\
\hline
\end{tabular}
\end{table}
\begin{table}[htb]
\centering
\caption{Explicit description of function Bases in Table 5.2}\label{table:5:3}
\begin{tabular}{cll}
\hline\hline
\\
Bases       &\small{$\rho^{W}\left(t,x\right)$ with $t\in[T_{j-1},T_{j})$}  &\small{$\rho^{\mathcal{P}}\left(t,y_{k},x\right)$ with $t\in\left[T_{j-1},T_{j}\right)$, $1\leq k\leq \mathcal{K}$} \\ \hline
\\
Bases 1     &$\{1\}$                            &$\{1\}$ \\
Bases 2     &$\{1,\; x,\; x^{2},\;x^{3}\}$      &$\{1,\; x,\; x^{2},\;x^{3}\}$ \\
Bases 3     &\scriptsize{$\{1,\;C^{BS}\left(t,x;T\right),\; \bigl(C^{BS}\left(t,x;T\right)\bigr)^{2}\}$}  & \scriptsize{$\{1,\;C^{BS}\left(t,x;T\right),\; \bigl(C^{BS}\left(t,x;T\right)\bigr)^{2}\}$} \\
Bases 4     &\scriptsize{$\{1,\frac{\partial C^{BS}\left(t,x;T_{j+1}\right)}{\partial x}x, \frac{\partial C^{BS}\left(t,x;T\right)}{\partial x} x \}$}   &\tiny{$\{1, C^{BS}\left(t,x\cdot e^{y_{k}};T_{j+1}\right)-C^{BS}\left(t,x;T_{j+1}\right),C^{BS}\left(t,x\cdot e^{y_{k}};T\right)-C^{BS}\left(t,x;T\right)\}$}\\
\hline
\end{tabular}
\end{table}
\\
Another interesting experiment has been conducted to investigate the individual performance of the two terms in \eqref{eq:4:4} since each term individually is a well-defined true martingale (adapted to the filtration $\{\mathcal{F}_{T_{j}}:j=0,...,\mathcal{J}\}$). Specifically, if
\begin{equation*}
\begin{split}
\hat{M}^{\pi,\mathcal{A}}_{{T_j}}&= \sum\limits_{0 \le {t_l} < {T_j}} {{{\hat \phi }^{\pi ,\mathcal{A}}}\left( {{t_l},{X^{\bar{\pi},\bar{\mathcal{A}}}_{{t_l}}}} \right)} \left( {{\Delta ^\pi }{W_l}} \right) + \sum\limits_{0 \le {t_l} < {T_j}} {\sum\limits_{k = 1}^\mathcal{K} {{{\hat \psi }^{\pi ,\mathcal{A}}}\left( {{t_l},y_{k},{X^{\bar{\pi},\bar{\mathcal{A}}}_{{t_l}}}} \right)\tilde{P}\left( {\left[ {{t_l},{t_{l + 1}}} \right] \times {A_k}} \right)} }\\
&=(Term \;1) \; + \; (Term \; 2),
\end{split}
\end{equation*}
then both Term 1 and Term 2 are martingales adapted to the filtration $\{\mathcal{F}_{T_{j}}:j=0,...,\mathcal{J}\}$. Therefore both of them will induce true upper bounds. Results in Table \ref{eq:5:4} show that taking out either one of these two terms yields significantly poorer upper bounds with much worse confidence intervals, which implies that the effort we have spent on the regression coefficients $\hat{\alpha}$ and $\hat{\beta}$, and the construction of the martingale $M^{\pi,\mathcal{A}}$ is necessary and time-worthy.
\begin{table}[htb]
\centering
\caption{Upper Bounds by one term in the True Martingale}\label{table:5:4}
\begin{tabular}{cccccc}
\hline\hline
$\lambda$  &$X_{0}$     &Term 1              &Term 2              &Complete Martingale            \\ \hline
\\
$1$        &$36$        &$6.863\pm0.059$      &$7.930\pm0.073$    &$5.970\pm0.031$    \\
$1$        &$40$        &$4.450\pm0.056$      &$5.184\pm0.072$    &$3.910\pm0.033$    \\
$1$        &$44$        &$2.750\pm0.050$      &$3.125\pm0.064$    &$2.443\pm0.028$    \\
\\
$3$        &$36$        &$10.101\pm0.099$     &$9.304\pm0.070$    &$7.899\pm0.030$    \\
$3$        &$40$        &$7.776\pm0.103$      &$7.047\pm0.070$    &$5.996\pm0.047$    \\
$3$        &$44$        &$5.747\pm0.098$      &$5.244\pm0.066$    &$4.480\pm0.044$    \\
\hline
\end{tabular}
\end{table}
\section{Conclusion and Future directions}
In this paper, we propose a true martingale algorithm (T-M algorithm) to fast compute the upper bounds on Bermudan option price under the jump-diffusion models, as an alternative approach for the classic A-B algorithm proposed by \cite{andersen:2004}, especially when the computational budget is limited. The theoretical analysis of our algorithm proves and the numerical results verify that our algorithm generates stable and tight upper bounds with significant reduction of computational effort. Moreover, we explore the structure of the optimal dual martingale for the dual problem and provide an intuitive understanding towards the construction of good approximations of the optimal dual martingale over the space of all adapted martingales.\\
\\
Furthermore, from the information relaxation point of view (see \cite{brown:2010}), we can gain an intuitive understanding towards the structure of the optimal penalty function. It inspires us to construct good penalty functions over the space of ``feasible penalty functions'' for general dynamic programming problems, which is still an open area to explore (see \cite{ye2012parameterized} for some initial exploration).
\bibliographystyle{ormsv080}
\bibliography{Zhou-Bibtex}

\end{document}